\documentclass[a4paper,11pt]{article}
\usepackage[utf8]{inputenc}
\usepackage{amssymb}
\usepackage{amsfonts,amsmath,amsthm}
\usepackage{mathtools,braket}
\usepackage[margin=1in]{geometry}
\usepackage{enumerate,enumitem}
\usepackage{multirow, multicol}
\usepackage{thmtools, thm-restate}
\usepackage[plain]{algorithm}
\usepackage{algpseudocode,color}

\usepackage[backend=biber,style=alphabetic,maxalphanames=4,maxcitenames=2,maxbibnames=99,natbib=true]{biblatex}
\usepackage{varwidth}
\usepackage{url}
\usepackage{soul} 
\usepackage[breaklinks,hyperfootnotes=false]{hyperref}
\hypersetup{
	colorlinks   = true, 
	urlcolor     = blue, 
	linkcolor    = blue, 
	citecolor   = blue 
}

\newtheorem{theorem}{Theorem}
\newtheorem{lemma}[theorem]{Lemma}
\newtheorem{proposition}[theorem]{Proposition}

\newtheorem{problem}{Problem}
\newtheorem{claim}{Claim}
\newtheorem{definition}{Definition}

\newcommand{\red}[1]{{\color{red}#1}}
\newcommand{\blue}[1]{{\color{blue}#1}}

\newcommand{\abs}[1]{\left|#1\right|}

\DeclareMathOperator*{\argmin}{arg\,min}

\algnewcommand\algprocedure{\textbf{Procedure:}}
\algnewcommand\Procedurename{\item[\underline{\algprocedure}]}
\algnewcommand\algmain{\textbf{Main:}}
\algnewcommand\Main{\item[\underline{\algmain}]}
\algnewcommand\algorithmicinput{\textbf{Input:}}
\algnewcommand\Input{\item[\algorithmicinput]}
\algnewcommand\algorithmicoutput{\textbf{Output:}}
\algnewcommand\Output{\item[\algorithmicoutput]}

\urldef{\email}\path|{vijay.menon, kate.larson}@uwaterloo.ca|

\allowdisplaybreaks

 \title{Robust and Approximately Stable Marriages under Partial Information}
\author{
	Vijay Menon\footnote{David R.\ Cheriton School of Computer Science, University of Waterloo.\ \email} \\
	\newline
	\and
	Kate Larson\footnotemark[1]
}

\date{}

\begin{document}

 \maketitle

\begin{abstract}
 We study the stable marriage problem in the partial information setting where the agents, although they have an underlying true strict linear order, are allowed to specify partial 
orders either because their true orders are unknown to them or they are unwilling to completely disclose the same. Specifically, we focus on the case where the agents are 
allowed to submit strict weak orders and we try to address the following questions from the perspective of a market-designer: \textit{i)} How can a designer generate 
matchings that are robust---in the sense that they are ``good'' with respect to the underlying unknown true orders? \textit{ii)} What is the trade-off between the amount of 
missing information and the ``quality'' of solution one can get? With the goal of resolving these questions through a simple and prior-free approach, we suggest looking at 
matchings that minimize the maximum number of blocking pairs with respect to all the possible underlying true orders as a measure of ``goodness'' or ``quality'', and 
subsequently 
provide results on finding such matchings.

In particular, we first restrict our attention to matchings that have to be stable with respect to at least one of the completions (i.e., weakly-stable matchings) and show 
that in this case arbitrarily filling-in the missing information and computing the resulting stable matching can give a non-trivial approximation factor (i.e., $o(n^2)$) 
for our problem in certain cases. We complement this result by showing that, even under severe restrictions on the preferences of the agents, the factor obtained is 
asymptotically tight in many cases. We then investigate a special case, where only agents on one side provide strict weak orders and all the missing 
information is at the bottom of their preference orders, and show that in this special case the negative result mentioned above can be circumvented in order to get a much 
better approximation factor; this result, too, is tight in many cases. Finally, we move away from the restriction on weakly-stable matchings and show a general hardness of 
approximation result and also discuss one possible approach that can lead us to a near-tight approximation bound.
\end{abstract}

\section{Introduction}
Two-sided matching markets have numerous applications, e.g., in matching students to dormitories ({i.e.}, Stable Roommates problem (SR) \citep{irv85}), residents to 
hospitals ({i.e.}, Hospital-Resident problem (HR) \citep{man08}) etc., and hence are ubiquitous in practice. Perhaps unsurprisingly, then, this line of research has 
received much 
attention, with plenty of work done on investigating numerous problems like SR and HR, and their many variations (we refer the reader to the excellent books by 
Gusfield and Irving \cite{gus89} 
and Manlove \cite{man13} for a survey on two-sided matching problems). The focus of this paper, too, is on one such problem---one that is perhaps the most widely-studied, 
but 
yet the simplest---called  the \textit{Stable Marriage problem} (SM), first introduced by Gale and Shapley \cite{gale62}. In SM we are given two disjoint sets (colloquially 
referred 
to as the set of men and women) and each agent in one set specifies a strict linear order over the agents in the other set, and the aim is to find a \textit{stable matching}, 
i.e., a matching where there is no man-woman pair such that each of them prefers the other over their partner in the matching. (Such a pair, if it exists, is called a 
\textit{blocking pair}.) 

While the assumption that the agents will be able to specify strict linear orders is not unreasonable in small markets, in general, as the markets get larger, it may not be 
feasible for an agent to determine a complete ordering over all the alternatives. Furthermore, there may arise situations where agents are simply unwilling to provide strict 
total orders due to, say, privacy concerns.
Thus, it is natural for a designer to allow agents the flexibility to specify partial orders, and so in this 
paper we assume that the agents submit strict weak orders\footnote{All our negative results naturally hold for the case when the agents are allowed to specify strict partial 
orders. As for our positive results, most of them can be extended for general partial orders, although the resulting bounds will be worse.} (i.e, strict partial orders where 
incomparability is transitive) that are consistent with their underlying true strict linear orders. Although the issue of partially specified preferences has received 
attention previously, we argue that certain aspects have not been addressed sufficiently. In particular, the common approach to the question of what 
constitutes a ``good'' matching in such a setting has been to either work with stable matchings that arise as a result of an arbitrary linear extension of the submitted 
partial orders (these are known as \textit{weakly-stable matchings}) or to look at something known as \textit{super-stable matchings}, which are matchings that are stable 
with 
respect to all the possible linear extensions of the submitted partial orders \citep{irv94,ras14}. In the case of the former, one key issue is that we often do not 
really know how ``good'' a particular weakly-stable matching is with the respect to the underlying true orders of the agents, and in the case of the latter they often do not 
exist. Furthermore, we believe that it is in the interest of the market-designer to understand how robust or ``good" a matching is with respect to the underlying 
true orders of the agents, for, if otherwise, issues relating to instability and market unravelling can arise since the matching that is output by a mechanism can be  
arbitrarily bad with respect to these true orders.
Hence, in this paper we propose to move away from the extremes of working with either arbitrary weakly-stable matchings or super-stable matchings, and to find a middle-ground 
when it comes to working with partial preference information. To this end, we aim to 
answer two questions from the perspective of a market-designer: \emph{i)} How should one handle partial information so as to be able to provide some guarantees with respect 
to the underlying true preference orders? \emph{ii)} What is the trade-off between the amount of missing information and the quality of a matching that one can achieve? We 
discuss our proposal in more detail in the following sections.

\subsection{How does one work with partial information?} 

When agents do not submit full preference orderings, there are several possible ways to cope with the missing information. For instance, one approach that immediately comes 
to 
mind is to assume that there exists some underlying distribution from which the agents' true preferences are drawn, and then use this information to find a ``good'' 
matching---which is, say, the one with the least number of blocking pairs in expectation. However, the success of such an approach crucially depends on having 
access to information about the underlying preference distributions which may not always be available. Therefore, in this paper we make no assumptions on the underlying 
preference distributions and instead adopt a prior-free and absolute-worst-case approach where we assume that any of the linear extensions of the given strict partial orders 
can be the underlying true order, and we aim to provide solutions that \emph{perform well} with respect to all of them. We note that similar worst-case 
approaches have been looked at previously, for instance, by Chiesa et al. \cite{chi12} in the context of auctions. 

The objective we concern ourselves with here is that of minimizing the number of blocking pairs, which is well-defined and has been considered previously in the context of 
matching problems (for instance, see \citep{abr05,biro10}). In particular, for a given instance $\mathcal{I}$ our aim is to return a matching $\mathcal{M}_{opt}$ that has the 
best worst case---i.e., a matching that has the minimum maximum `regret' after one realises the true underlying preference orders. (We refer to $\mathcal{M}_{opt}$ as the 
minimax optimal solution.) More precisely, let $\mathcal{I} = (p_U, p_W)$ denote an instance, where $p_U = \{p_{u_1}, \cdots, p_{u_n}\}$, $p_W = \{p_{w_1}, \cdots, 
p_{w_n}\}$, 
$U = \{u_i\}_{i\in\{1,2,\cdots,n\}}$ and $W=\{w_i\}_{i\in\{1,2,\cdots,n\}}$ are the set of men and women respectively, and $p_i$ is the strict partial order submitted by 
agent $i$. Additionally, let $C(p_i)$ denote the set of linear extensions of $p_i$, $C$ be the Cartesian product of the $C(p_i)$s, i.e., $C = \bigtimes_{i \in U \cup 
W} C(p_{i})$, $bp(\mathcal{M}, c)$ denote the set of blocking pairs that are associated with the matching $\mathcal{M}$ according to some linear extension $c \in C$, and $S$ 
denote the set of all possible matchings. Then the matching $\mathcal{M}_{opt}$ that we are interested in is defined as $\mathcal{M}_{opt} = \argmin_{\mathcal{M} \in S} 
\max_{c \in C} \abs{bp(\mathcal{M}, c)}$.

While we are aware of just one work by Drummond and Boutilier \cite{dru13} who consider the minimax regret approach in the context of stable matchings (they consider it 
mainly in the context of 
preference elicitation; see Section~\ref{sec:relatedwork} for more details), the approach, in general, is perhaps reminiscent, for 
instance, of the works of Hyafil and Boutilier \cite{hya04} and Lu and Boutilier \cite{lu11} who looked at the minimax regret solution criterion in the context of mechanism 
design for games with 
type uncertainty and preference elicitation in voting protocols, respectively. 

\textbf{Remark:} In the usual definition of a minimax regret solution, there is a second term which measures the `regret' as a result of choosing a particular solution. That 
is, in the definition above, it would usually be $\mathcal{M}_{opt} = \argmin_{\mathcal{M} \in S} \max_{c \in C} \abs{bp(\mathcal{M}, c)} - \abs{bp(\mathcal{M}_{c}, c)}$, 
where $\mathcal{M}_{c}$ is the optimal matching (with respect to the objective function $\abs{bp()}$) for the linear extension $c$. We do not include this in the definition 
above because $\abs{bp(\mathcal{M}_{c}, c)} = 0$ as every instance of the marriage problem with linear orders has a stable solution (which by definition has zero blocking 
pairs).  Additionally, the literature on stable matchings uses the term ``regret'' to denote the maximum cost associated with a stable matching, where the cost of a matching 
for an agent is the rank of its partner in the matching and the maximum is taken over all the agents (for instance, see \citep{man02}). However, here the term regret is used 
in the context of the minimax regret solution criterion.  

\subsection{How does one measure the amount of missing information?} \label{sec:missinfo}
For the purposes of understanding the trade-off between the amount of missing information and the ``quality'' of solution one can achieve, we need a way to measure the amount 
of missing information in a given instance. There are many possible ways to do this, however in this paper we adopt the following. For a given instance 
$\mathcal{I}$, the amount of missing information, $\delta$, is the fraction of pairwise comparisons one cannot infer from the given strict partial orders. That is, we know 
that if 
every agent submits a strict linear order over $n$ alternatives, then we can infer $\binom{n}{2}$ comparisons from it. Now, instead, if an agent $i$ submits a strict partial 
order $p_i$, 
then we denote by $\delta_i$ the fraction of these $\binom{n}{2}$ comparisons one cannot infer from $p_i$ (this is the ``missing information'' in $p_i$). Our $\delta$ here is 
equal to $\frac{1}{2n}\sum_{i\in U \cup W} \delta_i$. Although, given a strict partial order $p_i$, it is straightforward to calculate $\delta_i$, we will nevertheless assume 
throughout that $\delta$ is part of the input. Hence, our definition of an instance will be modified the following way to include the parameter for missing information: 
$\mathcal{I} = (\delta, p_U, p_W)$. 

\textbf{Remark:} $\delta = 0$ denotes the case when all the preferences are strict linear orders. Also, for an instance with $n$ agents on each side, the 
least value of $\delta$ when the amount of missing information is non-zero is $\frac{1}{2n}\frac{1}{\binom{n}{2}}$ (this happens in the case where there is only one agent 
with 
just one pairwise comparison missing). However, despite this, in the interest of readability, we sometimes just write statements of the form ``for all $\delta > 0$''. Such 
statements need to be understood as being true for only realizable or valid values of $\delta$ that are greater than zero.   

\subsection{Our Contributions}

The focus of our work is on computing the minimax optimal matching, i.e., a matching that, when given an instance 
$\mathcal{I}$, minimizes the maximum number of blocking pairs with respect to all the possible linear extensions (see Section~\ref{sec:prob} for a formal definition of the 
problem). Towards this end, we make the following contributions: 

\begin{itemize}
\item We formally define the problem and show that, interestingly, the problem under consideration is equivalent to the problem of finding a 
matching that has the minimum number of \textit{super-blocking pairs} (i.e., man-woman pairs where each of them weakly-prefers the other over their current partners).

\item While an optimal answer to our question might involve matchings that have man-woman pairs such that each of them strictly prefers the other over their partners, 
we start by focusing our investigation on matchings that do not have such pairs. Given the fact that any matching with no such pairs are weakly-stable, through this setting 
we address the question ``given an instance, can we find a weakly-stable matching that performs well, in terms of minimizing the number of blocking pairs, 
with respect to all the linear extensions of the given strict partial orders?'' We show that by arbitrarily filling-in the missing information and computing the resulting 
stable 
matching, one can obtain a non-trivial approximation factor (i.e., one that is ${o}(n^2)$) for our problem for many values of $\delta$.  We complement this result by showing 
that, even under severe restrictions on the preferences of the agents, the factor obtained is asymptotically tight in many cases. 

\item By assuming a special structure on the agents' preferences---one where strict weak orders are specified by just agents on one side and all the missing information is at 
the bottom of their preference orders---we show that one can obtain a $\mathcal{O}(n)$-approximation algorithm for our problem. 
The proof of the same is via finding a 2-approximation for another problem (see Problem~\ref{prob:min-del}) that might be of independent interest. 

\item In Section~\ref{sec:beyondws} we remove the restriction to weakly-stable matchings and show a general hardness of approximation result for our problem. Following 
this, we discuss one possible approach that can lead to a near-tight approximation guarantee for the same. 
\end{itemize}

\subsection{Related Work} \label{sec:relatedwork}
There has recently been a number of papers that have looked at problems relating to missing preference information or uncertainty in preferences in the context of matching. 

Drummond and Boutilier \cite{dru13} used the minimax regret solution criterion in order to drive preference elicitation strategies for matching problems. While they discussed 
computing robust matchings subject to a minimax regret solution criteria, their focus was on providing an NP-completeness result and heuristic preference elicitation 
strategies for refining the missing information. In contrast,  in addition to focusing on understanding the exact trade-offs between the amount of missing information and the 
solution ``quality'', we concern ourselves with arriving at approximation algorithms for computing such robust matchings.

Rastegari et al. \cite{ras14} studied a partial information setting in labour markets. However, again, the focus of this paper was different than ours. They looked at 
pervasive-employer-optimal matchings, which are matchings that are employer-optimal (see \citep{ras14} for the definitions) with respect to all the underlying linear 
extensions. In addition, they also discussed how to identify, in polynomial time, if a matching is employer-optimal with respect to some linear extension. 

Recent work by Aziz et al. \cite{aziz16} looked at the stable matching problem in settings where there is uncertainty 
about the preferences of the agents. They considered three different models of uncertainty and primarily studied the complexity of computing the 
stability probability of a given matching and the question of finding a matching that will have the highest probability of being stable. In contrast to their work, in this 
paper we do not make any underlying distributional assumptions about the preferences of the agents and instead take an absolute worst-case approach, which in turn implies 
that our results hold irrespective of the underlying distribution on the completions. 

Finally, we also briefly mention another line of research which deals with partial information settings and goes by the name of \textit{interview 
minimization} (see, for instance, \citep{ras13,dru14}). One of the main goals in this line of work is to come with a matching that is stable (and possibly satisfying some 
other desirable property) by conducting as few `interviews' (which in turn helps the agents in refining their preferences) as possible. We view this work as an interesting, 
orthogonal, direction from the one we pursue in this paper. 

\section{Preliminaries} \label{sec:prelims}
Let $U$ and $W$ be two disjoint sets. The sets $U$ and $W$ are colloquially referred to as the set of men and women, respectively, 
and $|U| = |W| = n$. We assume that each agent in $U$ and $W$ has a true strict linear order (i.e., a ranking without ties) over the agents in the other set, but this 
strict linear order may be unknown to the agents or they may be unwilling to completely disclose the same. Hence, each agent in $U$ and $W$ specifies a strict partial order 
over the agents 
in the other set (which we refer to as their \textit{preference order}) that is consistent with their underlying true orders, and $p_U$ and $p_W$, respectively, denote the 
collective preference orders of all the men and women. For a strict partial order $p_i$ associated with agent $i$, we denote the set of linear extensions associated with 
$p_i$ by 
$C(p_i)$ and denote by $C$ the Cartesian product of the $C(p_i)$s, i.e., $C = \bigtimes_{i \in U \cup W} C(p_{i})$. We refer to the set $C$ as ``the set of all completions'' 
where the term \textit{completion} refers to an element in $C$. Also, throughout, we denote strict preferences by $\succ$ and use $\succeq$ to denote the relation 
`weakly-prefers'. So, for instance, we say that an agent $c$ strictly prefers $a$ to $b$ and denote this by $a \succ_{c} b$ and use  $a \succeq_{c} b$ to denote that  either 
$c$ strictly prefers $a$ to $b$ or finds them incomparable. As mentioned in the introduction, we restrict our attention to the case when the strict partial orders submitted by 
the agents are strict weak orders over the set of agents in the other set. 

\textbf{Remark:} Strict weak orders are defined to be strict partial orders where incomparability is transitive. Hence, although the term \textit{tie} 
is used to mean 
indifference, it is convenient to think of strict weak orders as rankings with ties. Therefore, throughout this paper, whenever we say that agent $c$ finds $a$ and $b$ to be 
tied, we mean that $c$ finds $a$ and $b$ to be incomparable. Additionally, we will use the terms ties and incomparabilities interchangeably.   

An instance $\mathcal{I}$ of the stable marriage problem (SM) is defined as $\mathcal{I} = (\delta, p_U, p_w)$, where $\delta$ denotes the 
amount of missing information in that instance and this in turn, as defined in Section~\ref{sec:missinfo}, is the average number of pairwise comparisons that are missing from 
the instance, and $p_U$ and $p_W$ are as defined above. Given an instance $\mathcal{I}$, the aim is usually to come up with a matching $\mathcal{M}$---which in turn is a set 
of disjoint pairs $(m, w)$, where $m \in U$ and $w \in W$---that is stable. There are different notions of stability that have been proposed and below we define two of them 
that 
are relevant to our paper: \emph{i)} weak-stability and \emph{ii)} super-stability. However, before we look at their definitions we introduce the following terminology that 
will be used 
throughout this paper. (Note that in the definitions below we implicitly assume that in any matching $\mathcal{M}$ all the agents are matched. This is so because of the 
standard assumption that is made in the literature on SM (i.e., the stable marriage problem where every agent has a strict linear order over all the agents in the other set) 
that 
an agent always prefers to be matched to some agent than to remain unmatched.)

\begin{definition}[blocking pair/obvious blocking pair]
 Given an instance $\mathcal{I}$ and a matching $\mathcal{M}$ associated with $\mathcal{I}$, $(m, w)$ is said to be a blocking pair associated with $\mathcal{M}$ if $w 
\succ_{m} \mathcal{M}(m)$ and $m \succ_{w} \mathcal{M}(w)$. The term blocking pair is usually used in situations where the preferences of the agents are strict linear orders, 
so in cases where the preferences of the agents have missing information, we refer to such a pair as an obvious blocking pair.
\end{definition}

\begin{definition}[super-blocking pair]
 Given an instance $\mathcal{I}$ where the agents submit partial preference orders and a matching $\mathcal{M}$ associated with $\mathcal{I}$, we say that $(m, w)$ is a 
super-blocking pair associated with $\mathcal{M}$ if $w \succeq_{m} \mathcal{M}(m)$ and $m \succeq_{w} \mathcal{M}(w)$.
\end{definition}

Given the definitions above we can now define weak-stability and super-stability. 

\begin{definition}[weakly-stable matching]
 Given an instance $\mathcal{I}$ and matching $\mathcal{M}$ associated with $\mathcal{I}$, $\mathcal{M}$ is so said to be weakly-stable with respect to $\mathcal{I}$ if it 
does not have any obvious blocking pairs. When the preferences of the agents are strict linear orders, such a matching is just referred to as a stable matching. 
\end{definition}

\begin{definition}[super-stable matching]
  Given an instance $\mathcal{I}$ and matching $\mathcal{M}$ associated with $\mathcal{I}$, $\mathcal{M}$ is so said to be super-stable with respect to $\mathcal{I}$ if it 
does not have any super-blocking pairs.
\end{definition}

\subsection{What problems do we consider?} \label{sec:prob}
As mentioned in the introduction, we are interested in finding the minimax optimal matching where the objective is to minimize the number of blocking 
pairs, i.e., to find, from the set $S$ of all possible matchings, a matching that has the minimum maximum number of blocking pairs with respect to all the completions. 
This is formally defined below. 

\begin{problem}[$\delta$-minimax-matching] Given a $\delta \in [0, 1]$ and an instance $\mathcal{I} = (\delta', p_U, p_W)$, where $\delta' \leq \delta$ is the amount of 
missing information and $p_U, p_W$ are the preferences submitted by men and women respectively, compute $\mathcal{M}_{opt}$ where $\mathcal{M}_{opt} = \argmin_{\mathcal{M} 
\in S} \max_{c \in C} \abs{bp(\mathcal{M}, c)}$. 
\end{problem}

Although the problem defined above is our main focus, for the rest of this paper we will be talking in terms of the following problem which concerns itself with finding an 
approximately super-stable matching (i.e., a super-stable matching with the minimum number of super-blocking pairs). As we will see below, the reason we do the same is 
because both the problems are equivalent. 

\begin{problem}[$\delta$-min-bp-super-stable-matching] Given a $\delta \in [0, 1]$ and an instance $\mathcal{I} = (\delta', p_U, \allowbreak p_W)$, where $\delta' \leq 
\delta$ is the amount of missing information and $p_U, p_W$ are the preferences submitted by men and women respectively, compute $\mathcal{M}^{SS}_{opt}$ where 
$\mathcal{M}^{SS}_{opt} = \argmin_{\mathcal{M} \in S} \abs{\text{super-bp}(\mathcal{M})}$ and $\text{super-bp}(\mathcal{M})$ is the set of super-blocking pairs associated with 
$\mathcal{M}$ for the instance $\mathcal{I}$.
\end{problem}

Below we show that both the problems described above are equivalent. However, before that we prove the following lemma.

\begin{lemma} \label{lem:prob13}
Let $\mathcal{M}$ be a matching associated with some instance $\mathcal{I} = (\delta, p_U, p_W)$, $\alpha$ denote the maximum number of blocking pairs associated with 
$\mathcal{M}$ for any completion of $\mathcal{I}$, and $\beta$ denote the number of super-blocking pairs associated with $\mathcal{M}$ for the 
instance $\mathcal{I}$. Then, $\alpha = \beta$.
\end{lemma}

\begin{proof}
 First, it is easy to see that if there are $\alpha$ blocking pairs associated with $\mathcal{M}$ for a completion, then there are at least as many super-blocking pairs 
associated with $\mathcal{M}$. Therefore, $\alpha \leq \beta$. 

Next, we will show that if $\beta$ is the number of super-blocking pairs associated with $\mathcal{M}$, then $\mathcal{I}$ has at least one completion such that it has 
$\beta$ number of blocking pairs associated with $\mathcal{M}$. To see this, for each $m_i \in U$ and for each $w_j \in W$ such that $(m_i, w_j)$ is a super-blocking 
pair, do the following:
\begin{itemize}
 \item if $m_i$ finds $w_j$ incomparable to $\mathcal{M}(m_i)$, then construct a new partial order $p'_{m_i}$ for $m_i$ such that it is the same as $p_{m_i}$ except for the 
fact that in $p'_{m_i}$ we have that $m_i$ strictly prefers $w_j$ over $\mathcal{M}(m_i)$.
\item if $w_j$ finds $m_i$ incomparable to $\mathcal{M}(w_j)$, then construct a new partial order $p'_{w_j}$ for $w_j$ such that it is the same as $p_{w_j}$ except for the 
fact that in $p'_{w_j}$ we have that $w_j$ strictly prefers $m_i$ over $\mathcal{M}(w_j)$. 
\end{itemize}

Once the above steps are done, if there still exists any agent whose preference order is partial, then complete it arbitrarily. Now, consider this instance $\mathcal{I}'$ 
that is obtained. Then, again, it is easy to see that every $(m_i, w_j)$ which was a super-blocking pair associated with $\mathcal{M}$ in $\mathcal{I}$ forms a blocking pair 
in $\mathcal{M}$ with respect to $\mathcal{I}'$. Therefore, 
this completion has $\beta$ blocking pairs, and since the maximum number of blocking pairs in any completion is $\alpha$, we have that $\beta \leq \alpha$. Combining this 
with the case above, we have that $\alpha = \beta$.
\end{proof}

Given the lemma above, we can now prove our theorem.

\begin{theorem} \label{thm:prob13}
 For any $\delta \in [0, 1]$, the $\delta$-minimax-matching and $\delta$-min-bp-super-stable-matching problems are equivalent.
\end{theorem}

\begin{proof}
To see this, let $\mathcal{I} = (\delta', p_U, p_W)$ be some instance, where $\delta' \leq \delta$, and $\mathcal{M}$ be some matching associated with $\mathcal{I}$. We 
show that $\mathcal{M}$ is an 
optimal 
solution for the $\delta$-minimax-matching problem if and only if it is an optimal solution for the $\delta$-min-bp-super-stable-matching problem. 

$(\implies)$ Let us suppose that $\mathcal{M}$ is not an optimal solution for the $\delta$-min-bp-super-stable-matching problem. This implies that there exists some other 
$\mathcal{M}'$ such that $\abs{\text{super-bp}(\mathcal{M}')} \allowbreak < \abs{\text{super-bp}(\mathcal{M})}$. However, from 
Lemma~\ref{lem:prob13} we know that the maximum number of blocking pairs associated with $\mathcal{M}'$ for any completion with respect to $\mathcal{I}$ is equal to 
$\abs{\text{super-bp}(\mathcal{M}')}$, which in turn contradicts the fact that $\mathcal{M}$ was optimal for the $\delta$-minimax-matching problem. 

$(\impliedby)$ We can prove this analogously.
\end{proof}
 
For the rest of this paper, we assume that we are always dealing with instances which do not have a super-stable matching as this can be checked in polynomial-time 
\citep[Theorem 3.4]{irv94}. So, now, in the context of the $\delta$-min-bp-super-stable-matching problem, it is easy to show that if the number of super-blocking pairs $k$ in 
the optimal solution is a constant, then we can solve it in polynomial-time. We state this in the theorem below. Later, in Section~\ref{sec:beyondws}, we will see that the 
problem is NP-hard, even to approximate.

\begin{theorem} \label{thm:const-opt1}
 An exact solution to the $\delta$-min-bp-super-stable-matching problem can be computed in $\mathcal{O}(n^{2(k+1)})$ time, where $k$ is the number of 
super-blocking pairs in the optimal solution. 
\end{theorem}

\begin{proof}
 We will describe the algorithm below whose main idea is based on the following observation. 

 For an instance $\mathcal{I}$, consider its optimal solution $\mathcal{M}_{opt}$ and let the $k$ super-blocking pairs associated with $\mathcal{M}_{opt}$ be $B = \{(m_1, 
w_1), \cdots, (m_k, w_k)\}$. Next, for each such pair $(m_i, w_i)$, put $m_i$ ($w_i$) at the end of $w_i$'s ($m_i$'s) preference list (i.e., make every other man (woman), 
except those involved in another blocking pair with $w_i$ ($m_i$), rank better than $m_i$ ($w_i$)). If either of them are involved in multiple blocking pairs in $B$, then 
make those partners as incomparable at the end of the preference list. Let us call the new instance $\mathcal{I}'$. Notice that $\mathcal{M}_{opt}$ is super-stable with 
respect to $\mathcal{I}'$ as the pairs in $B$ are no longer blocking and no new blocking pairs are created because of our manipulations to the preference list.  

Given the above observation, we can now describe the exponential algorithm.
\begin{itemize}
 \item Initially $j=1$. Given a $j$, try out every possible set of pairs of size $j$ to see if they are the right blocking pairs.
 \item For each set generated in the previous step, modify the original instance $\mathcal{I}$ to $\mathcal{I}'$ as described above and see if $\mathcal{I}'$ has a 
super-stable matching (this can be done in polynomial time). If yes, then return the super-stable matching as that is the solution. Otherwise, if none of the sets of size $j$ 
result in a ``yes'', then go back to step 1 and try again with the next value of $j$.  
\end{itemize}

Now, it is easy to see that we end up with the optimal solution this way since we try all possible sets of blocking pairs. As for the time, we know that for each $j$ we have 
at most $(n^2)^j$ choices of sets and for each set we need at most $2n^2$ time to do the necessary manipulations to the instance and to check for super-stability. 
Hence, the total time required is $\sum_{j=1}^{k} 2n^{2j + 2} = \mathcal{O}(n^{2(k+1)})$.
\end{proof}

\section{Investigating weakly-stable matchings} \label{sec:ws}

In this section we focus on situations where obvious blocking pairs are not permitted in the final matching. In particular, we explore the space of weakly-stable matchings 
and ask whether it is possible to find weakly-stable matchings that also provide good approximations to the $\delta$-min-bp-super-stable-matching problem (and thus 
the $\delta$-minimax-matching problem).

\subsection{\texorpdfstring{Approximating $\delta$-min-bp-super-stable-matching with weakly-stable matchings}{}} \label{sec:ws-1}
It has previously been established that a matching is weakly-stable if and only if it is stable with respect to at least one completion \citep[Section 1.2]{man02}. Therefore, 
given this result, one immediate question that arises in the context of approximating the $\delta$-min-bp-super-stable-matching problem is ``what if we just fill in the 
missing information arbitrarily and then compute a stable matching associated with such a completion?'' This is the question we consider here, and we show that 
weakly-stable matchings do give a non-trivial (i.e., one that is $o(n^2)$, as any matching has only $\mathcal{O}(n^2)$ super-blocking pairs) approximation bound for our 
problem for certain values of $\delta$. The proof of the following theorem is through a simple application of the Cauchy-Schwarz inequality.

\begin{theorem}\label{thm:WO}
 For any $\delta > 0$ and an instance $\mathcal{I} = (\delta', p_U, p_W)$ where $\delta' \leq \delta$, any weakly-stable matching with respect to $\mathcal{I}$ gives 
an $\mathcal{O}\left(\min\left\{n^3\delta, n^2\sqrt{\delta}\right\}\right)$-approximation for the $\delta$-min-bp-super-stable-matching problem.
\end{theorem}

\begin{proof}
 Let $\mathcal{M}$ be a weakly-stable matching associated with $\mathcal{I}$. By the definition of weakly-stable matchings we know that $\mathcal{M}$ does 
not have any obvious blocking pairs. This implies that for every super-blocking pair $(m, w)$ associated with $\mathcal{M}$, either $m$ finds $w$ incomparable to his 
partner 
$\mathcal{M}(m)$ or $w$ finds $m$ incomparable to her partner $\mathcal{M}(w)$. If it is the former then we refer to the super-blocking pair $(m, w)$ as one that is 
associated with $m$ and if not we say that it is associated with $w$. Next, let us suppose that there are $d$ agents who have a blocking pair associated with them and 
let $b_i$ denote the number of super-blocking pairs associated with agent $i$. So, now, the number of super-blocking pairs, $\abs{\text{super-bp}(\mathcal{M})}$, associated 
with $\mathcal{M}$ can be written as 
 \begin{align} \label{ws:eq1}
  \abs{\text{super-bp}(\mathcal{M})} &= \sum_{i = 1}^{d} b_{i} \leq \sum_{i = 1}^{d} (\ell_{i} - 1) = \sum_{i = 1}^{d} \ell_{i} - d ,
 \end{align}
 where $\ell_i$ refers to the length of largest tie associated with agent $i$ and the inequality follows from the definition of an association of a super-blocking pair with 
an agent.
 
 Additionally, for each $i \in \{1, \cdots, d\}$, we know that at least $\binom{\ell_i}{2}$ pairwise comparisons are missing with respect to $i$ (since $i$ has a 
tie of length $\ell_i$). Therefore, using the Cauchy-Schwarz inequality, we have that 
 \begin{align} \label{ws:eq2}
  \sum_{i=1}^{d} \binom{\ell_{i}}{2} 
  &= \frac{1}{2}\sum_{i=1}^{d} \ell_{i}^2 - \ell_{i} 
  \geq \frac{1}{2} \left[\frac{1}{d}\left(\sum_{i=1}^{d} \ell_{i}\right)^2 -  \sum_{i=1}^{d} \ell_{i}\right]. 
 \end{align}
 
 Also, since the total amount of missing information $\delta'$ in the instance $\mathcal{I}$ is less than or equal to $\delta$ and since each $\ell_i \geq 2$ (as 
it is a weak order and a tie, if it exists, is of length at least 2) we 
have that 
 \begin{align} \label{eq3}
  &d \leq \sum_{i=1}^{d} (\ell_i - 1) \leq \sum_{i=1}^{d} \binom{\ell_{i}}{2} \leq \delta (2n) \binom{n}{2} \leq \delta n^3.
 \end{align}
 
 Now, using Equation~\ref{eq3} and the fact that $d$ is also upper-bounded by $2n$ (as there are only $2n$ agents in the instance), we have that $d \leq 
\min\{2n, \delta n^3\}$. 
Therefore, using Equation~\ref{ws:eq2} and again using the fact that $\delta$ is the maximum amount of missing information, we have,
 \begin{align} 
  &\frac{1}{2} \left[\frac{1}{d}\left(\sum_{i=1}^{d} \ell_{i}\right)^2 -  \sum_{i=1}^{d} \ell_{i}\right] \leq \sum_{i=1}^{d} \binom{\ell_{i}}{2} \leq \delta (2n) 
\binom{n}{2}.
\end{align}
This in turn implies that if we solve for $\sum_{i=1}^{d} \ell_i$, we have,
\begin{align*}
  \sum_{i=1}^{d} \ell_{i} \leq \frac{1}{2} \left(d + \sqrt{d^2 + 8dn^2(n-1)\delta}\right) < d + \sqrt{d^2 + 8dn^2(n-1)\delta}.
\end{align*}
So, now, we can use the fact that $d \leq \min\{2n, \delta n^3\}$ to see that
\begin{align*}
 & \sum_{i=1}^{d} \ell_{i} - d \leq \min\{4n^3\delta, 5n^2\sqrt{\delta}\}. 
\end{align*}
 
 Finally, this along with Equation~\ref{ws:eq1} gives our result since the number of super-blocking pairs in the optimal solution is at least 1 (since, as mentioned in 
Section~\ref{sec:prob}, we are only considering instances that do not have a super-stable matching).
\end{proof}

\subsection{Can we do better when restricted to weakly-stable matchings?} \label{sec:ws-2}
While Theorem~\ref{thm:WO} established an approximation factor for the $\delta$-min-bp-super-stable-matching problem when considering only weakly-stable matchings, it was 
simply based on arbitrarily filling-in the missing information. Therefore, there remains the question as to whether one can be clever about handling the missing information 
and as a result obtain improved approximation bounds. In this section we consider this question and show that for many values of $\delta$ the approximation factor obtained 
in Theorem~\ref{thm:WO} is asymptotically the best one can achieve when restricted to weakly-stable matchings.

\begin{theorem} \label{thm:ws-tight}
For any $\delta \in [\frac{16}{n^2}, \frac{1}{4}]$, if there exists an $\alpha$-approximation algorithm for $\delta$-min-bp-super-stable-matching that always returns a 
matching that is 
weakly-stable, then $\alpha \in \Omega\left(n^2\sqrt{\delta}\right)$. Moreover, this result is true even if we allow only one side to specify ties and also insist that 
all the ties need to be at the top of the preference order.  
\end{theorem}

\begin{proof}
 At a very high-level, the key idea in this proof is to create an instance $\mathcal{I}$ such that if we insist on there being no obvious blocking pairs, then this results in 
some kind of a ``cascading effect'', thus in turn causing a very sharp blow-up in the number of super-blocking pairs. With this intuition, we first construct an 
instance $\mathcal{I}$ as shown in Figure~\ref{fig1}, where ties appear only on the women's side. Furthermore, we define the following:  
\begin{itemize}
 \item $y = \frac{n\sqrt{\delta}}{2}$,  $z = \frac{n}{2y}$ (for simplicity we assume that $y$ and $z$ are integers; we can appropriately modify the proof if that is not the 
case)
 \item $b_j = \frac{n}{2} + jy + 1, \forall j \in [0, \cdots z]$, 
 $B_i = \{b_{i-1}, \cdots, b_{i} - 1\}, \forall i \in [1, \cdots z]$
 \item$F = \{1, \cdots, \frac{n}{2}\}, S = \{\frac{n}{2} + 1, \cdots, n\}$
 \item$W_X\,(M_X):$ {for some set $X$, place all the women (men) with index in $X$ in the increasing order of their indices}
 \item$W^T_X\,(M^T_X):$ {for some set $X$, place all the women (men) with index in $X$ as tied}
 \item$[ \cdots ]:$ {place all the remaining alternatives in some strict order}.
\end{itemize}

\begin{figure}[t!] 
{\scriptsize
\noindent\hrulefill\\[0.5ex]
\begin{minipage}{0.49\textwidth}
\centering \underline{Men}
\begin{alignat*}{1}
 m_1 &: w_1 \succ W_{B_1}\succ \cdots \succ W_{B_z} \succ W_{F\setminus\{1\}}\\
 m_2 &: w_1 \succ w_2 \succ [ \cdots ]\\
 m_3 &: w_2 \succ w_3 \succ W_{F\setminus\{2, 3\}} \succ W_S\\
 m_4 &: w_2 \succ w_4 \succ W_{F\setminus\{2, 4\}} \succ W_S\\[-1ex]
 &\hspace{20mm}\vdots\\[-1ex]
 m_{\frac{n}{2}} &: w_2 \succ w_{\frac{n}{2}} \succ W_{F\setminus\{2, \frac{n}{2}\}} \succ W_S\\
 m_{b_0} &: w_1 \succ w_{b_0} \succ W_{B_1 \setminus \{b_0\}} \succ W_{S\setminus B_1} \succ W_{F\setminus\{1\}}\\[-1ex]
 &\hspace{20mm}\vdots\\[-1ex]
 m_{b_1 - 1} &: w_1 \succ w_{b_1 - 1} \succ W_{B_1 \setminus \{b_1 - 1\}} \succ W_{S\setminus B_1} \succ W_{F\setminus\{1\}}\\
 m_{b_1} &: w_1 \succ w_{b_1} \succ W_{B_2 \setminus \{b_1\}} \succ W_{S\setminus B_2} \succ W_{F\setminus\{1\}}\\[-1ex]
 &\hspace{20mm}\vdots\\[-1ex]
 m_{b_2 - 1} &: w_1 \succ w_{b_2 - 1} \succ W_{B_2 \setminus \{b_2 - 1\}} \succ W_{S\setminus B_2} \succ 
W_{F\setminus\{1\}}\\[-1ex]
 &\hspace{20mm}\vdots\\[-1ex]
   m_{b_{z-1}} &: w_1 \succ w_{b_{z-1}} \succ W_{B_z \setminus \{b_{z-1}\}} \succ W_{S\setminus B_z} \succ 
W_{F\setminus\{1\}}\\[-1ex]
 &\hspace{20mm}\vdots\\[-1ex]
 m_{b_{z} - 1} &: w_1 \succ w_{b_{z} - 1} \succ W_{B_z \setminus \{b_{z} - 1\}} \succ W_{S\setminus B_z} \succ 
W_{F\setminus\{1\}}
\end{alignat*} 
\end{minipage}
 \vrule{} 
\begin{minipage}{0.49\textwidth}
\centering \underline{Women}
\begin{alignat*}{1}
 w_1&: m_2 \succ m_1 \succ [ \cdots ]\\
 w_2&: m_2 \succ M_{(F\cup S) \setminus \{1\}} \succ m_1\\
 w_3&: m_1 \succ m_3 \succ [ \cdots ]\\
 w_4&: m_1 \succ m_4 \succ [ \cdots ]\\[-1ex]
 &\hspace{20mm}\vdots\\[-1ex]
 w_{\frac{n}{2}}&: m_1 \succ m_{\frac{n}{2}} \succ [ \cdots ]\\
 w_{b_0}&: M^T_{B_1 \setminus \{b_0\}} \succ M_{S\setminus B_1} \succ m_1 \succ m_{b_0} \succ M_{F\setminus\{1\}}\\[-1ex]
 &\hspace{20mm}\vdots\\[-1ex]
 w_{b_1 - 1} &:  M^T_{B_1 \setminus \{b_1 - 1\}} \succ M_{S\setminus B_1} \succ m_1 \succ m_{b_1-1} \succ M_{F\setminus\{1\}}\\
 w_{b_1} &:  M^T_{B_1 \setminus \{b_1\}} \succ M_{S\setminus B_1} \succ m_1 \succ m_{b_1} \succ M_{F\setminus\{1\}}\\[-1ex]
 &\hspace{20mm}\vdots\\[-1ex]
 w_{b_2 - 1} &:  M^T_{B_1 \setminus \{b_2-1\}} \succ M_{S\setminus B_1} \succ m_1 \succ m_{b_2-1} \succ M_{F\setminus\{1\}}\\[-1ex]
 &\hspace{20mm}\vdots\\[-1ex]
 w_{b_{z-1}} &:  M^T_{B_1 \setminus \{b_{z-1}\}} \succ M_{S\setminus B_1} \succ m_1 \succ m_{b_{z-1}} \succ M_{F\setminus\{1\}}\\[-1ex]
 &\hspace{20mm}\vdots\\[-1ex]
 w_{b_{z} - 1} &: M^T_{B_1 \setminus \{b_z-1\}} \succ M_{S\setminus B_1} \succ m_1 \succ m_{b_z-1} \succ M_{F\setminus\{1\}}
\end{alignat*}
\end{minipage}
\hrule 
\caption{\small The instance $\mathcal{I}$ that is used in the proof of Theorem~\ref{thm:ws-tight}}
\label{fig1}
}
\end{figure} 

Next, we will show that all the weakly-stable matchings associated with $\mathcal{I}$ have $\mathcal{O}\left(n^2\sqrt{\delta}\right)$ super-blocking pairs, whereas the 
optimal solution has exactly one super-blocking pair. To do this, first note that the optimal solution $\mathcal{M}_{opt}$ associated with the instance is
 $\mathcal{M}_{opt} = \{(m_1, w_1), (m_2, w_2), \allowbreak \cdots \allowbreak, (m_n, w_n)\}$,
where $(m_2, w_1)$ is the only super-blocking pair (and it is an obvious blocking pair). Also, 
it can be verified that the total amount of missing information in $\mathcal{I}$ is at most $\delta$. So, next, we prove the following claim.

\begin{claim}
 If $\mathcal{M}$ is a weakly-stable matching associated with the instance $\mathcal{I}$, then $\forall i \in \{\frac{n}{2}+1, \cdots, n\}, \mathcal{M}(m_i) \neq w_i$. 
\end{claim}
\begin{proof}
First, note that in any weakly-stable matching $m_2$ will always be matched to $w_1$ as otherwise it will result in an obvious blocking pair. Next, let us suppose that 
there 
exists an $i \in \{\frac{n}{2}+1, \cdots, n\}$ such that $\mathcal{M}(m_i) = w_i$. Now, we will consider the following two cases and show that in both the cases this is 
impossible. 
\item{\textbf{Case 1.} $m_1$ is matched to a woman $w \in W_{F\setminus\{1\}}$ in $\mathcal{M}$:} In this case one can see that $(m_1, w_i)$ forms an obvious blocking 
pair.  
\item{\textbf{Case 2.} $m_1$ is matched to a woman $w \in W_S$ in $\mathcal{M}$:} Note that if this is the case, then there is at least one $j \in S$ such that $m_j$ is 
matched with a 
woman $w \in W_F$. Now, notice that $(m_j, w_i)$ forms an obvious blocking pair.
%
\end{proof}

Given the claim above, consider a man $m$ whose index is in some block $B_j$ and let his index value be $k$. From the way the preferences are defined, it is easy to see 
that 
in any weakly-stable matching, $m$ will be matched to a woman $w$ whose index lies in the same block $B_j$ (because otherwise it will result in an obvious blocking 
pair). At the same time, from the claim above we know that this $w$'s index is not $k$. Now, let us consider the woman $w_{b_j-1}$ who is the woman with the highest index 
value in $B_j$ and let $m'$ denote the man who is matched to  $w_{b_j-1}$ in a weakly-stable matching. From the observation above we know that $m'$ has an index value in 
$B_j$. Additionally, given the way the preferences are defined for $m'$ and using the fact that any woman $w_{p}$ such that $p \in B_j$ finds all the men in 
$B_j\setminus\{p\}$ to be incomparable, one can see that $m'$ forms $(|B_j| - 2)$ super-blocking pairs (with all the women in $B_j$ except $w_{b_j-1}$ and the one with the 
same index value as $m'$). Also, by using the same argument again, but with respect to $w_{b_j-2}$, we can show that partner of $w_{b_j-2}$ in the matching forms at least 
$(|B_j| - 3)$ super-blocking pairs (with all women except $w_{b_j-2}$, $w_{b_j-1}$, and the one with the same index). Continuing this way we see that each block $B_j$ 
contributes $\mathcal{O}(|B_j|^2)$ super-blocking pairs. And so, since there are $z$ blocks and $|B_j| = y$ for all $j$, we have that there are 
$\mathcal{O}(n^2\sqrt{\delta})$ super-blocking pairs in any weakly-stable matching.
\end{proof}

\subsection{\texorpdfstring{The case of one-sided top-truncated preferences: An $\mathcal{O}(n)$ approximation algorithm for $\delta$-min-bp-super-stable-matching}{}}
\label{sec:ws-onesided}
Although Theorem~\ref{thm:ws-tight} is an inherently negative result, in this section we consider an interesting restriction on the preferences of the agents and show how 
this 
negative result can be circumvented. In particular, we consider the case where only agents on one side are allowed to specify ties and all the ties 
need to be at the bottom. Such a restriction has been looked at previously in the context of matching problems 
and as noted by Irving and Manlove \cite{irv08} is one that appears in practise in the Scottish Foundation Allocation Scheme (SFAS). Additionally, restricting ties to only at 
the bottom 
models a very well-studied class of preferences known as top-truncated preferences, which has received considerable attention in the context of voting (see, for 
instance, \citep{bau12}). 

Top-truncated preferences model scenarios where an agent is certain about their most preferred choices, but is indifferent among the remaining 
ones or is unsure about 
them. More precisely, in our setting, the preference order submitted by, say, a woman $w$ is said to be a top-truncated order if it is a linear order over a 
subset of $U$ and the remaining men  are all considered to be incomparable by $w$. In this section we consider one-sided top-truncated 
preferences, i.e., where only men or women are allowed to specify top-truncated orders, and show an $\mathcal{O}(n)$-approximation algorithm for 
$\delta$-min-bp-super-stable-matching under this setting. (Without loss of generality we assume throughout that only the women submit strict weak orders.) Although 
arbitrarily filling-in the missing information and computing the resulting weakly-stable matching can lead to an $\mathcal{O}(n^2\sqrt{\delta})$-approximate matching even for 
this restricted case (see Appendix~\ref{sec:example} for an example), we will see that not all weakly-stable matchings are ``bad'' and that in fact the 
$\mathcal{O}(n)$-approximate matching we obtain is weakly-stable.

However, in order to arrive at this result, we first introduce the following problem which might be of independent interest. (To the best of our knowledge, this has not been 
previously considered in the literature.) Informally, in this problem we are given 
an instance $\mathcal{I}$ and are asked if we can delete some of the agents to ensure that the instance, when restricted to the remaining agents, will have a 
perfect super-stable matching. 

\begin{problem}[min-delete-super-stable-matching] \label {prob:min-del} 
Given an instance $\mathcal{I} = (\delta, p_U, p_W)$, where $\delta$ is the amount of missing information and $p_U, p_W$ are the preferences submitted by men and women 
respectively, compute the set $D$ of minimum cardinality such that the instance $\mathcal{I}_{-D} = (\delta_{-D}, p_{U\setminus D}, p_{W\setminus D})$, where $\delta_{-D} = 
\frac{1}{|(U \cup W) \setminus D|}\sum_{i \in (U \cup W) \setminus D} \delta_i$, has a perfect super-stable matching (i.e., every agent in $(U \cup W) 
\setminus D$ is matched in a super-stable matching). 
\end{problem}

Below we first show a 2-approximation for the min-delete-super-stable-matching problem when restricted to the case of one-sided top-truncated 
preferences. Subsequently, we use this result in order to get an $\mathcal{O}(n)$-approximation for our problem. However, before that, we introduce the following terminology 
which will be used throughout in this section.
\begin{itemize}
\item An instance $\mathcal{I}$ of the min-delete-super-stable-matching problem can also be thought of as the set of agents along with their preference lists. Initially 
for every agent this list has all the agents in the other set listed in some order. Now, during the course of our algorithm sometimes we use the operation ``delete$(a,b)$'' 
which removes agent $a$ from $b$'s list and $b$ from $a$'s. After such a deletion (or after a series of such deletions) our instance now refers to the set of agents along 
with their updated lists. 
 \item We say that a matching $\mathcal{M}$ is \textit{internally super-stable} with respect to an instance $\mathcal{I}$ if $\mathcal{M}$ is super-stable with respect to the 
instance that is obtained by only considering the matched agents in $\mathcal{M}$ (i.e., consider $\mathcal{I}$ and remove all the agents who 
are not matched in $\mathcal{M}$ from $\mathcal{I}$).
 \item We say that an instance $\mathcal{I}$ with no ties has an exposed rotation $\rho = (m_1, w_1), \allowbreak  \cdots, (m_r, w_r)$ if, in 
$\mathcal{I}$, $w_i$ is the first agent in $m_i$'s list and $w_{i+1}$ is the second agent in $m_i$'s list (here 
$(i+1)$ is done modulo $r$).
\end{itemize}


\begin{algorithm}[t!]
{\footnotesize
\begin{center}
\noindent\fbox{%
\begin{varwidth}{\dimexpr\linewidth-3\fboxsep-3\fboxrule\relax}
\begin{algorithmic}[1]
  \Procedurename proposeWith$(A, \mathcal{I})$
\State assign each agent $a \in A$ to be free
\While{some $a \in A$ is free} \label{line2}
    \State $b \leftarrow $ first agent on $a$'s list \label{line3} 
    \If{$b$ is already engaged to agent $p$ \&\& $b$ finds $p$ and $a$ incomparable}
      \State delete $(a,b)$       
    \Else 
      \If{$b$ is already engaged to agent $p$}
	\State assign $p$ to be free
      \EndIf
      \State assign $a$ and $b$ to be engaged
      \For{each agent $c$ in $b$'s list such that $a \succ_{b} c$}
	\State delete $(c, b)$      
   \EndFor
   \EndIf \label{line14}
 \EndWhile \label{line15}
 \For{each man $m$}
  \State $w \leftarrow$ first woman on $m$'s list 
  \If{there exists a man $m'$ such that $w$ finds $m$ and $m'$ incomparable}
    \State delete $(m', w)$ \label{line19}
  \EndIf
\EndFor \Comment{{\tiny deletions in this loop only happen once and results in the removal of  all the remaining ties}}
 \State return $\mathcal{I}$ \Comment{{\tiny this returns the updated lists}}

\Main
\Input a one-sided top-truncated instance $\mathcal{I} = (\delta, p_U, p_W)$
\State $\mathcal{I'} \leftarrow \text{proposeWith}(U, \mathcal{I})$ \label{line23}
\State $\mathcal{I'} \leftarrow \text{proposeWith}(W, \mathcal{I'})$ \label{line24}
\While{there exists some exposed rotation $(m_1, w_1), (m_2, w_2), \cdots, (m_r, w_r)$ in $\mathcal{I}'$} \label{mainWhile}
    \State delete $(m_i, w_i)$ for all $i \in \{1, \cdots, r\}$
    \State $\mathcal{I'} \leftarrow \text{proposeWith}(U, \mathcal{I'})$
    \State $\mathcal{I'} \leftarrow \text{proposeWith}(W, \mathcal{I'})$
\EndWhile \label{line29}
\State $\mathcal{M} \leftarrow$ for all men $m \in U$, match $m$ with the only woman in his list
\State construct $G=(V, E)$ where $V = U \cup W$, $(m,w)\in E$ if $(m, w)$ is a super-blocking pair in $\mathcal{M}$ w.r.t.\ $\mathcal{I}$
\State $D \leftarrow$ minimum vertex cover of $G$ \label{lineVC}
\For{each $a \in D$} \label{line33}
  \State $D \leftarrow D \cup \mathcal{M}(a)$ 
\EndFor \label{line35}
\State return $(D, \mathcal{M})$
\end{algorithmic}  
\end{varwidth}
}
\end{center}
}
\caption{For the case of one-sided top-truncated preferences, the set $D$ returned by the algorithm is a 2-approximation for the 
min-delete-super-stable-matching problem and the matching $\mathcal{M}$ returned is an $\mathcal{O}(n)$-approximation for $\delta$-min-bp-super-stable-matching}
\label{algo1}
\end{algorithm}

\begin{proposition} \label{prop:2approx} 
 Algorithm~\ref{algo1} is a polynomial-time 2-approximation algorithm for the min-delete-super-stable-matching problem when restricted to the case of one-sided top-truncated 
preferences.
\end{proposition}

\begin{proof}
 The main idea for Algorithm~\ref{algo1} is inspired by the work of \citet{tan90} who looked at the problem of finding the maximum internally stable matching for the stable 
roommates problem (which is equivalent to the problem of finding the minimum number of agents to delete so that the rest of the agents will have a stable matching when the 
instance is just restricted to themselves). Informally, at a very high level, the key idea in \citeauthor{tan90}'s algorithm was to show that some of the entries in each 
agent's list can be deleted by running the proposal-rejection sequence like in Gale-Shapley algorithm and through rotation eliminations, while at the same time maintaining at 
least one solution of the maximum size. As we will see below, this is essentially what we do here as well, adapting this idea as necessary for our case when 
there are ties on one side but only at the bottom.

Before we go on to the main lemmas, let us suppose that $\mathcal{I} = (\delta, p_U, p_W)$ is an arbitrary instance of the min-delete-super-stable-matching problem when 
restricted to the case of one-sided top-truncated preferences, where $\delta$ is the amount of missing information, $p_U, p_W$ are the preference orders submitted by the men 
and women, respectively, and $|U|=|W|=n$. Also, let $D_{opt}$ be the optimal solution for this instance. This in turn implies that we can form a perfect and internally 
super-stable matching of size $k = n - \frac{D_{opt}}{2}$, and that in fact $k$ is the maximum size of any such matching (as otherwise $D_{opt}$ cannot be optimal). Next, for 
now, let us assume the correctness of the following lemmas (note that all the instances we talk about in this section are restricted to the case of one-sided top-truncated 
preferences). We will prove them later, in Sections~\ref{sec:prof12}, \ref{sec:prof13}, and~\ref{sec:prof14}.

\begin{lemma} \label{lem:proposeWith}
Let $\mathcal{I}_1$ denote some instance and $\mathcal{I}_2$ denote the instance returned by the procedure proposeWith$(A, \mathcal{I}_1$), where the set $A$ 
represents the proposing side. If there exists a matching of size $t$ in $\mathcal{I}_1$ that is internally super-stable with respect to $\mathcal{I}$, then there exists a 
matching of size $t$ in $\mathcal{I}_2$ that is internally super-stable with respect to $\mathcal{I}$.     
\end{lemma}

\begin{lemma} \label{lem:rotate}
 Let $\mathcal{I}_1$ denote some instance that does not contain any ties, $(m_1, w_1), (m_2, w_2), \allowbreak \cdots, \allowbreak (m_r, w_r)$ be a rotation that is exposed 
in 
$\mathcal{I}_1$, and $\mathcal{I}_2$ be the instance that is obtained by deleting the entries $(m_i, w_i)$ for all $i \in \{1, \cdots, r\}$ from $\mathcal{I}_1$. If there 
exists a matching of size $t$ in $\mathcal{I}_1$ that is internally super-stable with respect to $\mathcal{I}$, then there exists a matching of size $t$ in $\mathcal{I}_2$ 
that is internally super-stable with respect to $\mathcal{I}$. 
\end{lemma}

\begin{lemma} \label{lem:one-to-one}
 If $\mathcal{I}_1$ is an instance that does not contain any exposed rotation, then the list of every man in $\mathcal{I}_1$ has only one woman and vice versa. 
\end{lemma}

Given the above lemmas and given the fact that the instance $\mathcal{I}$ has an internally super-stable matching of size $k = n - \frac{D_{opt}}{2}$, we can 
start with the instance $\mathcal{I}$ and repeatedly apply Lemma~\ref{lem:proposeWith} and see that the instance $\hat{\mathcal{I}}$ that we obtain 
after line~\ref{line24} of Algorithm~\ref{algo1} has a matching of size $k$ that is internally super-stable matching with respect to $\mathcal{I}$. Next, starting with the 
instance $\hat{\mathcal{I}}$, we can again repeatedly apply 
Lemma~\ref{lem:rotate} and see that the instance $\mathcal{I}'$ that remains after 
line~\ref{line29} of Algorithm~\ref{algo1} also has a matching of size $k$ that is internally super-stable matching with respect to $\mathcal{I}$. 
Additionally, from Lemma~\ref{lem:one-to-one} we know that in this instance each man has only one woman in his list and vice versa. So, now, consider the matching 
$\mathcal{M}$ that can be obtained by matching each man with the only woman in his list. Next, consider the bipartite graph $G=(V,E)$ where $V = U \cup W$ and $(m, w) \in E$ 
if $(m, w)$ is a super-blocking pair with respect to the original instance $\mathcal{I}$, and consider a minimum vertex cover $C$ of $G$. We show below that $k \leq n - 
\frac{|C|}{2}$. 

Suppose this is false and that $\mathcal{I}$ has an internally super-stable matching of size greater than $n-\frac{|C|}{2}$. Now, from the discussion above, we know that this 
implies 
$\mathcal{I}'$ also has a matching of size greater than $n-\frac{|C|}{2}$ such that it is internally super-stable with respect to $\mathcal{I}$. This in turn implies that we 
can remove 
less than $|C|$ agents and have a matching that is internally super-stable with respect to $\mathcal{I}$. That is, we can remove less than $|C|$ agents and also at the same 
time ensure that for every $(m, w) \in E$, this matching has only one of $m$ or $w$ matched in it, for if otherwise it will not be internally super-stable with respect to 
$\mathcal{I}$. However, this implies that $|C|$ is not the size of the minimum vertex cover of $G$, and hence we have a contradiction.

Now, to get our approximation bound, consider the set $D$ that is returned by the algorithm. We know that $D \leq 2 |C| \leq 4(n-k) = 2D_{opt}$, where the first 
inequality arises because of lines~\ref{line33}-\ref{line35} in Algorithm~\ref{algo1} and the second inequality 
uses the observation above that $k \leq n - \frac{|C|}{2}$. Finally, it is easy to see that all the steps  can be done in polynomial-time (it is well-known through the 
K\H onig's Theorem that one can find a minimum vertex cover of a bipartite graph in polynomial-time). 
\end{proof}

Given Proposition~\ref{prop:2approx}, we can now prove the following theorem.

\begin{theorem} \label{thm:napprox} 
 For any $\delta > 0$, Algorithm~\ref{algo1} is a polynomial-time $\mathcal{O}(n)$-approximation algorithm for the $\delta$-min-bp-super-stable-matching problem when 
restricted to the case of one-sided top-truncated preferences. Moreover, the $\mathcal{O}(n)$-approximate matching that is returned is also weakly-stable.  
\end{theorem}

\begin{proof}
Consider an arbitrary instance $\mathcal{I}$ of the $\delta$-min-bp-super-stable-matching problem when restricted to the case of 
one-sided top-truncated preferences. Let $\mathcal{M}_{opt}$ be the optimal solution associated with 
$\mathcal{I}$. Next, consider the same instance $\mathcal{I}$ for the min-delete-super-stable-matching problem and let us consider the matching  $\mathcal{M}$ that is 
returned 
by Algorithm~\ref{algo1} for this instance. Also, let $D_{opt}$ be the optimal solution of the min-delete-super-stable-matching problem for the instance $\mathcal{I}$ and $D$ 
be the set that is returned by Algorithm~\ref{algo1}. (We can assume throughout that $D_{opt} \geq 1$, for if otherwise this implies that it has a super-stable matching, 
and as mentioned in Section~\ref{sec:prob} we do not consider such instances.) First, it is easy to see that this matching is 
weakly-stable (because in every instance that results after the first proposal-rejection sequence (which is in line 23), a matching that is formed by matching each man with 
the first woman on his list will be weakly-stable). Second, note that we can rewrite $\mathcal{M}$ as $\mathcal{M} = \mathcal{M}_1 \cup \mathcal{M}_2$, where  $\mathcal{M}_1 
= \mathcal{M} \setminus \{ \cup_{a \in D} (a, \mathcal{M}(a))\}$ and $\mathcal{M}_2 = \cup_{a \in D} (a, \mathcal{M}(a))$. Now, if $S_1$ ($S_2$) denotes the number of 
super-blocking pairs associated with men in $\mathcal{M}_1$ ($\mathcal{M}_2$), then
\begin{align} \label{app:eq1}
 \abs{\text{super-bp}(\mathcal{M})} &= S_1 + S_2 \nonumber \\
  &\leq \left(n-\frac{|D|}{2}\right)\cdot\frac{|D|}{2} + \frac{|D|}{2}\cdot n \nonumber\\
 &\leq n\cdot{|D|}  \nonumber \\
 &\leq2n\cdot|D_{opt}|,
 \end{align}
 where the second step is using the fact that the men in $\mathcal{M}_1$ can form at most $\frac{|D|}{2}$ super-blocking pairs with women outside of $\mathcal{M}_1$  and the 
men in $\mathcal{M}_2$ can in the worst case form a blocking pair with all the women, and 
the last step is using the fact that $D_{opt} \geq 1$ and that $|D| \leq 2 \cdot |D_{opt}|$, which we know is true by Proposition~\ref{prop:2approx}.
 
 Also, if $\mathcal{M}_{opt}$ is the optimal solution for the $\delta$-min-bp-super-stable-matching problem, then we know that  
 \begin{align}\label{app:eq2}
  \abs{\text{super-bp}(\mathcal{M}_{opt})} \geq \frac{|D_{opt}|}{2},
 \end{align}
 as otherwise one can delete all the men who are involved in super-blocking pairs in $\mathcal{M}_{opt}$ and their corresponding partners and get a 
super-stable matching on the remaining agents.
 
 Finally, using Equations~\ref{app:eq1} and~\ref{app:eq2} we have our theorem.
\end{proof} 

\subsubsection{Proof of Lemma~\ref{lem:proposeWith}} \label{sec:prof12}
Let us first prove the following claims. However, before that we introduce the following terminology. When the procedure proposeWith$(A,\mathcal{I}_1)$ is executed, for 
every run of the while loop in line~\ref{line2} with respect to an agent $a$, we can track the instance that is currently being used. That is, initially we have the instance 
$\mathcal{I}_1$ and this is referred to as the instance that is ``currently being used'' by the agent $a_1$ where $a_1$ is the first agent with respect whom the while loop is 
executed. Now, after the first run of the while loop (w.r.t.\ $a_1$) we have an updated instance (because of some delete operations that happened in 
lines~\ref{line3}-\ref{line14}), say, $\mathcal{I}_{curr}$. Therefore, the next time the while loop is run with respect to some agent $a_2$, this 
is the instance that is ``currently being used'' with respect to $a_2$. We use this terminology in the following claim.

\begin{claim} \label{clm:freeagent}
 Let $a \in A$ be an agent who is assigned to be `free' in the procedure proposeWith$(A,\mathcal{I}_1)$, $\mathcal{I}_{curr}$ denote the instance that is currently being used 
with respect to $a$, and $\mathcal{I}'_1$ be the instance obtained by running lines~\ref{line3}-\ref{line14} with respect to agent $a$. If there exists a matching of size 
$t$ in $\mathcal{I}_{curr}$ that is internally super-stable with respect to $\mathcal{I}$, then there exists a matching of size $t$ in $\mathcal{I}'_1$ that is internally 
super-stable with respect to $\mathcal{I}$.   
\end{claim}

\begin{proof}
First, note that every time the procedure proposeWith() is called with the set $A$, the agents in this set do not have any ties. This is because, when proposeWith() is called 
with the set of women $W$ for the first time in line~\ref{line24}, all ties are already broken due to the first call of proposeWith() with the set of men in 
line~\ref{line23}. Therefore, in all the arguments below we do not need to concern ourselves with issues that can arise as a result of the agents in $A$ having 
ties. 

Now, to prove this claim we consider the following two cases separately.
\begin{enumerate}
 \item when $b$, the first agent in $a$'s list, is either engaged to $p$, but prefers $a$ to $p$, or is not engaged currently
 \item when $b$ is engaged to $p \in A$ and finds $p$ and $a$ incomparable
\end{enumerate}

\item{\textbf{Case 1}:} In this case the only change that happens to the instance $\mathcal{I}_{curr}$ are due to deletions of the form $(c, b)$ where $a \succ_b c$. Let the 
resulting instance be $\mathcal{I}'_1$. So, in order to prove that $\mathcal{I}'_1$ has a matching of size $t$ that is internally super-stable with respect to $\mathcal{I}$, 
we just need to prove that there exists a matching of size $t$ in $\mathcal{I}_{curr}$ that is internally super-stable matching with respect to $\mathcal{I}$ and does not 
have 
any 
matches of the form $(c, b)$. Now, let us suppose that's not the case and that in every matching $\mathcal{M}$ of size $t$ in $\mathcal{I}_{curr}$ that is internally 
super-stable 
with respect to $\mathcal{I}$ there exists such a pair. This in turn implies that $a$ is unmatched in $\mathcal{M}$, for if otherwise $(a, b)$ will form an 
obvious-blocking pair. Hence, we can form the matching $\mathcal{M}' = (\mathcal{M} \setminus (c, b)) \cup (a, b)$. Note that $\mathcal{M}'$ is internally-stable with respect 
to $\mathcal{I}$ as this does not introduce any new super-blocking pair (since $b$ is the first agent on $a$'s list and $b$ does not block any other agent in $\mathcal{M}'$ 
as he prefers the new partner over his old partner $c$) and has the same size as $\mathcal{M}$. This in turn is a contradiction and hence we have that $\mathcal{I}'_1$ has a 
matching of size $t$ that is internally super-stable with respect to $\mathcal{I}$.

\item{\textbf{Case 2}:} In this case, the only change that happens to the instance $\mathcal{I}_{curr}$ is the deletion of $(a, b)$. Let the resulting instance be 
$\mathcal{I}'_1$. So, in order to prove that $\mathcal{I}'_1$ has a matching of size $t$ that is internally super-stable with respect to $\mathcal{I}$, 
we just need to prove that there exists a matching of size $t$ in $\mathcal{I}_{curr}$ that is internally super-stable matching with respect to $\mathcal{I}$ and does not 
have any matches of the form $(a, b)$. Suppose that's not the case and that every matching $\mathcal{M}$ of size $t$ in $\mathcal{I}_{curr}$ that is internally super-stable 
matching with respect to $\mathcal{I}$ has $(a, b)$. This in turn implies that $p$ is unmatched in $\mathcal{M}$, for if otherwise $(p, b)$ will form a 
super-blocking pair (since $b$ is the first agent on $p$'s list and $b$ finds $p$ and $a$ incomparable). Hence, we can now form the matching $\mathcal{M}' = (\mathcal{M} 
\setminus (a, b)) \cup (p, b)$, which is internally super-stable with respect to $\mathcal{I}$ as this does not introduce any new super-blocking pair (since i) $b$ is the 
first agent on $p$'s list and ii) $b$ does not block any other agent in $\mathcal{M}'$ as it finds $a$ and $b$ incomparable and so if he blocks someone now he also blocked 
them when $(a, b)$ was present, thus contradicting the fact that $\mathcal{M}$ was internally super-stable) and has the same size as $\mathcal{M}$. This in turn is a 
contradiction and hence, again, we have that $\mathcal{I}'_1$ has a matching of size $t$ that is internally super-stable with respect to $\mathcal{I}$.
\end{proof}

Next, we prove the following claim whose proof is omitted since it can be proved in a way that is similar to Case 2 in the proof of Claim~\ref{clm:freeagent}.

\begin{claim} \label{clm:tiedel}
 Let $m$ be a man, $w$ be the first woman on $m$'s list, $\mathcal{I}_3$ be some instance obtained after line~\ref{line15}, and $\mathcal{I}_4$ be the instance obtained 
after deleting each $(m', w)$ in line~\ref{line19}. If there exists a matching of size $t$ in $\mathcal{I}_3$ that is internally super-stable with respect to $\mathcal{I}$, 
then there exists a matching of size $t$ in $\mathcal{I}_4$ that is internally super-stable with respect to $\mathcal{I}$.      
\end{claim}

Given the two lemmas above, we are now ready to prove Lemma~\ref{lem:proposeWith}. First, starting with $\mathcal{I}_1$ and by repeatedly using Claim~\ref{clm:freeagent} with 
respect to each free agent $a \in A$, we can see that the instance $\mathcal{I}'_1$ that we obtain at the end of the while loop (i.e., line~\ref{line15}) has a matching of 
size $t$ that is internally super-stable matching  with respect to the initial instance $\mathcal{I}$. Second, starting with $\mathcal{I}'_1$ and by repeatedly 
using Claim~\ref{clm:tiedel} with respect to each man, we can see that the instance $\mathcal{I}_2$ that is returned from the procedure proposeWith$(A, \mathcal{I}_1)$ has a 
matching of size $t$ that is internally super-stable with respect to the instance $\mathcal{I}$. This in turn proves our lemma. \qed

\subsubsection{Proof of Lemma~\ref{lem:rotate}} \label{sec:prof13}

Let $\mathcal{M}$ be a matching of size $t$ in $\mathcal{I}_1$ that is internally super-stable with respect to $\mathcal{I}$. We will consider the following two cases 
separately and show that in each case there exists a matching of size $t$ in $\mathcal{I}_2$ that is internally super-stable with respect to $\mathcal{I}$.
\begin{enumerate}
 \item for all $i \in \{1, \cdots, r\}$, $(m_i, w_i) \in \mathcal{M}$
 \item there exists some $j \in \{1, \cdots, r\}$ such that $(m_j, w_j) \notin \mathcal{M}$
\end{enumerate}

{\textbf{Case 1}:} Consider the matching $\mathcal{M}'$ that is obtained by removing $(m_i, w_i)$ and instead adding $(m_i, w_{i+1})$ for all $i \in \{1, \cdots, r\}$ 
(here $i+1$ is done modulo $r$). Now, it is easy to see that this does not lead to any new internal super-blocking pairs with respect to $\mathcal{I}$. Hence, $\mathcal{M}'$ 
is internally super-stable with respect to $\mathcal{I}$, has the same size as $\mathcal{M}$ (as every agent matched in $\mathcal{M}$ is also matched in 
$\mathcal{M}'$), and all the matched pairs in $\mathcal{M}'$ are in $\mathcal{I}_2$. 


{\textbf{Case 2}:} First, note that if $(m_i, w_i) \notin \mathcal{M}$ for all $i \in \{1, \cdots, r\}$, then we are done as the only deleted entries in $\mathcal{I}_2$ 
are $(m_i, w_i)$ for all $i \in \{1, \cdots, r\}$. So this implies that we can, without loss of generality, consider the smallest $k \in \{1, \cdots, r\}$ such that $(m_j, 
w_j) \notin \mathcal{M}$ for all $j \in \{1, \cdots, k\}$, but $(m_{k+1}, w_{k+1}) \in \mathcal{M}$ (since there is at least one $j$ such that $(m_j, w_j) \notin 
\mathcal{M}$, this can always be done because we can re-index
the rotation so that the first element $(m_1, w_1)$ of the rotation is not in $\mathcal{M}$). Let $\mathcal{I}_2'$ be the instance that is formed by deleting the entries 
$(m_i, w_i)$ for all $i\in\{1, \cdots, k+1\}$. Below, we will show that $\mathcal{I}_2'$ satisfies the conditions of the lemma, i.e., we show that $\mathcal{I}_2'$ has a 
matching of size $t$ that is internally super-stable with respect to $\mathcal{I}$. And so once we have that, we can just repeat this argument until we get to the instance 
$\mathcal{I}_2$. 

To see why $\mathcal{I}_2'$ satisfies the conditions of the lemma, notice that $m_k$ needs to be unmatched in $\mathcal{M}$, for if otherwise then $(m_k, w_{k+1})$ will be 
an obvious blocking-pair. This in turn implies that we can construct another matching $\mathcal{M}'$ such that $\mathcal{M}'= (\mathcal{M} \setminus (m_{k+1}, w_{k+1})) \cup 
(m_{k}, w_{k+1})$, both of them have the same size, and all the matched pairs in $\mathcal{M}'$ are in $\mathcal{I}_2'$. Also, one can see that this does not lead to any new 
super-blocking pairs since i) $w_{k+1}$ improved and so will not be part of any new super-blocking pair and ii) $m_k$ does not form a super-blocking pair as it is matched to 
the agent who is second in his list and $w_k$ who is first in his list, if matched in $\mathcal{M}$, has to be matched to someone better (as $m_k$ is the last agent in 
$w_k$'s list).
\qed

\subsubsection{Proof of Lemma~\ref{lem:one-to-one}} \label{sec:prof14}

First, note that throughout an agent $a$ is in $b$'s list if and only if $b$ is in $a$'s list. Second, we prove the following claim. 

\begin{claim} \label{clm:one-one}
Let $\mathcal{I}_2$ be some initial instance, $\mathcal{I}_3$ be the instance that is obtained after running the procedure proposeWith() with the men's side proposing, and 
$\mathcal{I}_4$ be the instance that is obtained after running the procedure proposeWith() with the women's side proposing. For any two agents $a$ and $b$ in 
$\mathcal{I}_4$, if $a$ is the only agent in $b$'s list, then $b$ is the only agent in $a$'s list. 
\end{claim}

\begin{proof}
To prove our claim let us consider the following two cases in the instance $\mathcal{I}_3$. 
\begin{enumerate}
 \item there exists a woman $w_1$ such that $m_1$ is the only man on $w_1$'s list, but $m_1$ has at least one other woman other than $w_1$ in his list
 \item there exists a man $m_1$ such that $w_1$ is the only woman on $m_1$'s list, but $w_1$ has at least one other man other than $m_1$ in her list
\end{enumerate}

\item{\textbf{Case 1}:} Since we are looking at  $\mathcal{I}_3$, note that $w_1$ must be the first woman on $m_1$'s list, for if otherwise there is some other man, say, 
$m_k$ who is not engaged (because $w_1$ has only $m_1$ in her list). However, we know that this is not possible and so $w_1$ is the first woman on $m_1$'s list. 
Therefore, now, when we run the procedure proposeWith() with the women's side proposing, $w_1$ will propose to $m_1$ and as a result $m_1$ will delete all other women from 
his list. And so, in $\mathcal{I}_4$, $w_1$ is the only agent in $m_1$'s list and vice versa. 

\item{\textbf{Case 2}:} For this case, let us suppose that even after running the procedure proposeWith() with the women's side proposing, $w_1$ still has at least one other 
man other than $m_1$ in her list (if not, then we are already done). Let this man be $m_2$. This in turn implies that in the engagement relation that results, at least one of 
the women, say, $w_k$, is not engaged (because $m_1$ has only $w_1$ in his list). However, we know that this is not possible as this would imply that the engagement 
relation has an obvious blocking pair $(m_k, w_k)$, where $m_k$ is the last agent on $w_k$'s list (and we know this cannot happen since we are just running the same 
proposal-rejection sequence as in the Gale-Shapley algorithm).
\end{proof}

Given the observations above, let us assume for the sake of contradiction that there exists an instance $\mathcal{I}_1$ such that it does not have any exposed rotation 
but has at least one agent who has a list of size greater than one. Using Claim~\ref{clm:one-one} and the fact that an agent $a$ is in $b$'s list if and only if $b$ is 
in $a$'s list, we can assume without loss of generality that there exists a man $m_1$ such that he has at least two women in his list. Let $w_1$ and $w_2$ be the first and 
second woman, respectively, in $m_1$'s list. Since $\mathcal{I}_1$ is obtained after running the procedure proposeWith() twice, once each with the men's and women's side 
proposing, from Claim~\ref{clm:one-one} we know that $m_2$, who is the last man on $w_2$'s list, too has at least two women in his list as otherwise $w_2$ would also have 
just one man $m_2$ in her list. Let $w_3$ be the second woman in $m_2$'s list. Now, we can see that we can just inductively keep on applying the above argument and form the 
following pairs $\rho = (m_1, w_1), (m_2, w_2), (m_3, w_3), \cdots, (m_r, w_r)$ for some $r \in \{2, \cdots, n\}$, where $m_i$ is the last man on $w_i$'s list and $w_{i+1}$ 
is the second woman on $m_i$'s list. However, note that $\rho$ is an exposed rotation, and hence we have a contradiction. \qed

Before we end this section, we address one final question as to whether, for the class of one-sided top-truncated preferences, one can obtain a 
better approximation result if one continues to consider only weakly-stable matchings. In the theorem below we show that for $\delta \in 
\Omega(\frac{1}{n})$ Algorithm~\ref{algo1} is asymptotically the best one can do under this restriction. 

\begin{theorem} \label{thm:tight}
For $\delta \leq \frac{1}{2}$, if there exists an $\alpha$-approximation algorithm for $\delta$-min-bp-super-stable-matching that always 
returns a matching that is weakly-stable for the case of one-sided top-truncated preferences, then $\alpha \in \Omega\left(\min\left\{n^{\frac{3}{2}}\sqrt{\delta}, 
n\right\}\right)$. 
\end{theorem}

\begin{proof}
 \begin{figure}[t!] 
{\scriptsize
\noindent\hrulefill\\[0.5ex]
\begin{minipage}{0.47\textwidth}
\centering \underline{Men}
\begin{alignat*}{1}
 m_1 &: w_2 \succ w_1 \succ [\cdots]\\
 m_2 &: w_3 \succ w_1 \succ w_2 \succ [ \cdots ]\\
 m_3 &: w_4 \succ w_1 \succ w_3 \succ [ \cdots ]\\
  &\hspace{5mm}\vdots\\[-1ex]
 m_{n-y}&:  w_{n-y+1} \succ [\cdots]\\ 
 m_{n-y+1} &: w_{n-y+2} \succ w_1 \succ w_{n-y+1} \succ [ \cdots ]\\
 m_{n-y+2} &: w_{n-y+2} \succ w_1 \succ w_{2} \succ [ \cdots ]\\
 m_{n-y+3} &: w_{n-y+2} \succ w_1 \succ w_{n-y+3} \succ [ \cdots ]\\
 &\hspace{5mm}\vdots\\[-1ex]
 m_n &: w_{n-y+2} \succ w_1 \succ w_n \succ [ \cdots ]
\end{alignat*} 
\end{minipage}
\vrule{} 
\noindent
\begin{minipage}{0.47\textwidth}
\centering \underline{Women} 
\begin{alignat*}{1}
 w_1&: m_1 \succ \cdots \succ m_{n-y} \succ (m_{n-y+1}, \cdots, m_n)\\
 w_2&: m_1 \succ [\cdots]\\ 
 w_3&: m_2 \succ [\cdots]\\
 w_4&: m_3 \succ [\cdots]\\
  &\hspace{5mm}\vdots\\[-1ex]
 w_{n-y+1}&:m_{n-y} \succ [\cdots]\\
 w_{n-y+2}&:m_{n-y+1} \succ m_{n-y+2} \succ [\cdots]\\ 
 w_{n-y+3}&: m_{n-y+3} \succ [\cdots]\\
 w_{n-y+4}&: m_{n-y+4} \succ [\cdots]\\
 &\hspace{5mm}\vdots\\[-1ex]
 w_n&: m_{n} \succ [\cdots]\\[-1ex]
\end{alignat*}
\end{minipage}
\hrule
\caption{\small The instance $\mathcal{I}$ that is used in the proof of Theorem~\ref{thm:tight}}
\label{fig3}
}
\end{figure} 

To prove this, consider the instance $\mathcal{I}$ defined as in Figure~\ref{fig3}, where $y = \min\left\{\lfloor 2n^{\frac{3}{2}}\sqrt{\delta} \rfloor, \allowbreak 
n\right\}$. Below we 
will show that the statement of the theorem is true when $\delta \leq \frac{1}{2n}$, i.e., show that in this case $\alpha \in 
\Omega\left(n^{\frac{3}{2}}\sqrt{\delta}\right)$. For $\delta > \frac{1}{2n}$ it is then trivial to extend our instance by just having some of the other women (other than 
$w_1$) specify partial preferences.

First, it is easy to verify that this instance has at most $\delta$ amount of information missing. Second, one can see that the optimal solution, i.e., the matching with the 
minimum number of super-blocking pairs, for this instance is 
\begin{multline*}
 \mathcal{M}_{opt} = \big\{(m_1, w_1), \allowbreak (m_2, w_3), \cdots, \allowbreak (m_{n-y+1}, w_{n-y+2}), (m_{n-y+2}, w_{2}),\\ \allowbreak (m_{n-y+3}, w_{n-y+3}), 
\allowbreak \cdots, \allowbreak (m_n, w_n)\big\},
\end{multline*}
where $(m_1, w_2)$ is the only super-blocking pair.

Given the above observation, let us now consider an arbitrary matching $\mathcal{M}$ that is weakly-stable. Since $m_i$ prefers $w_{i+1}$ the most and vice versa for all $i 
\in \{1, \cdots, n-y+1\}$, we know that $\mathcal{M}(m_i) = w_{i+1}$. Hence, $\mathcal{M}(w_1) = m_k$, for some $k \in \{n-y+2, \cdots, n\}$. Additionally, we also know 
that for all $j \in \{n-y+2, \cdots, n\}$ such that $j \neq k$, $w_1 \succ_{m_j} \mathcal{M}(m_j)$, as none of these can men can be matched to $w_{n-y+2}$. This in turn 
implies that since $w_1$ finds $m_k$ and $m_j$ incomparable for $j \in \{n-y+2 \cdots, n\}$ such that $j \neq k$, $(m_j, w_1)$ is a super-blocking pair. Therefore, in 
any weakly-stable matching $\mathcal{M}$ we have $(y-2) \in \mathcal{O}(n^{\frac{3}{2}}\sqrt{\delta})$ super-blocking pairs. 
\end{proof}

\section{Beyond Weak-Stability} \label{sec:beyondws}
In the previous section we investigated weakly-stable matchings and we showed several results concerning this situation. Here we move away from this restriction and explore 
what happens when we do not place any restriction on the matchings. In particular, we begin this section by showing a general hardness of approximation result, 
and then follow it with a discussion on one possible approach that can lead to a near-tight approximation result.

\subsection{\texorpdfstring{Inapproximability result for $\delta$-min-bp-super-stable-matching}{}}
We show a hardness of approximation result for the $\delta$-min-bp-super-stable-matching problem through a gap-producing reduction from the Vertex Cover 
(VC) problem, which is a well-known NP-complete problem \citep{karp72}. In the VC problem, we are given a graph $G=(V,E)$, where $V = \{v_1, \cdots, v_k\}$, and a 
$k_0 \leq k$ and are asked if there exists a subset of the vertices with size less than or equal to $k_0$ such that it contains at least one endpoint of every edge. 
When given an instance $\mathcal{I}$ of VC, the key idea in the proof is to create an instance $\mathcal{I}'$ of $\delta$-min-bp-super-stable-matching such that if 
$\mathcal{I}$ is a ``yes'' instance of VC, then $\mathcal{I}'$ will have a very small number of super-blocking pairs, and if otherwise, then $\mathcal{I}'$ will have a large 
number of super-blocking pairs. 
 
 \begin{theorem} \label{thm:inapprox-minss}
 For any constant $\epsilon \in (0,1]$ and $\delta \in (0, 1)$, one cannot obtain a polynomial-time $(n\sqrt{\delta})^{1-\epsilon}$ approximation algorithm for 
the $\delta$-min-bp-super-stable-matching problem unless P = NP.  
\end{theorem}

\begin{proof}
 The proof here is similar to the one by \citet[Theorem 1]{hamada16}. The main difference is in the construction of the instance, which in our case is more 
involved; once we have that we can essentially use the same proof as the one by \citeauthor{hamada16} 

Given an instance $\mathcal{I} = (G=(V,E), k_0)$ of the VC problem, where $|V| = k$, we construct the following instance $\mathcal{I}'$ of the 
$\delta$-min-bp-super-stable-matching problem, where  
\begin{itemize}
 \item $d = \left\lceil  \frac{8}{\epsilon} \right\rceil$, $y = k^{d} + 1$, $z = \lceil \frac{1}{\sqrt{\delta}} \rceil$
 \item $M_{A_1} = \{m_1, \cdots, m_{k_0}\}$, $M_{A_2} = \{m_{k_0+1}, \cdots, m_{k}\}$, $W_A = \{w_1, \cdots, w_k\}$
 \item for every $i < j$ such that $(v_i, v_j) \in E$ and $c \in \{1, \cdots, z\}$, $S_{c}^{i,j} = \{s^{i,j}_{c, 1}, \cdots s^{i,j}_{c, y}\}$, $T_{c}^{i,j} = \{t^{i,j}_{c, 
1}, \cdots t^{i,j}_{c, y}\}$, $P_{c}^{i,j} = \{p^{i,j}_{c, 1}, \cdots p^{i,j}_{c, y}\}$, and $V_{c}^{i,j} = \{v^{i,j}_{c, 1}, \cdots v^{i,j}_{c, y}\}$  
 \item $S^{i,j} = \{S_1^{i,j}, \cdots S_z^{i,j}\}, T^{i,j} = \{T_1^{i,j}, \cdots T_z^{i,j}\}, P^{i,j} = \{P_1^{i,j}, \cdots P_z^{i,j}\}$, and $V^{i,j} = \{V_1^{i,j}, \cdots 
\allowbreak V_z^{i,j}\}$
 \item $M_A = M_{A_1} \cup M_{A_2}$, $S = \bigcup S^{i,j}$, $T = \bigcup T^{i,j}$, $P = \bigcup P^{i,j}$, and $V = \bigcup V^{i,j}$
 \item $U = M_A \cup S \cup P$ and $W = W_A \cup T \cup V$
 \item for each $i, j$ the preference orders of the agents are as given in Figure~\ref{fig2}. For an agent $a$, if $R_S$ appears in its preference list for some $R$ and 
$S$, then this implies that $a$ finds all the agents in $R_S$ as incomparable. Also, $[\cdots]$ denotes that the rest of the agents can be placed in any order.
\end{itemize}

\begin{figure}[t!] 
{\scriptsize
\noindent\hrulefill\\[0.5ex]
\begin{minipage}{0.45\textwidth}
\vspace{-2mm}
\centering \underline{Men}
\begin{alignat*}{1}
 m_i &: W_A \succ V^{i, j}_1 \succ [ \cdots ]\\
 s^{i,j}_{1,1} &:\red{t^{i,j}_{1, 1}} \succ w_i \succ \blue{t^{i,j}_{\frac{z}{2}+1,1}} \succ V^{i,j}_1 \succ [ \cdots ]\\
 s^{i,j}_{1,2} &:\red{t^{i,j}_{1, 2}} \succ w_i \succ \blue{t^{i,j}_{1,3}} \succ V^{i,j}_1 \succ [ \cdots ]\\[-0.5ex]
 &\vdots \\[-0.5ex]
 s^{i,j}_{1,y} &:\red{t^{i,j}_{1, y}} \succ w_i \succ \blue{t^{i,j}_{2,1}} \succ V^{i,j}_1 \succ [ \cdots ]\\
 s^{i,j}_{2,1} &:\red{t^{i,j}_{2, 1}} \succ w_i \succ \blue{t^{i,j}_{2,2}} \succ V^{i,j}_2 \succ [ \cdots ]\\[-0.5ex]
 & \vdots\\[-0.5ex] 
 s^{i,j}_{2,y} &:\red{t^{i,j}_{2, y}} \succ w_i \succ \blue{t_{^{i,j}_{3,1}}} \succ V^{i,j}_2 \succ [ \cdots ]\\[-0.5ex]
 &\vdots \\[-1ex]
  &\vdots \\[-1ex]
 s^{i,j}_{\frac{z}{2},y} &:\red{t^{i,j}_{\frac{z}{2},y}} \succ w_i \succ\blue{t^{i,j}_{1,1}} \succ V^{i,j}_{\frac{z}{2}} \succ [ \cdots ]\\
 s^{i,j}_{\frac{z}{2}+1,1} &: \blue{t^{i,j}_{1,2}} \succ w_j \succ \red{t^{i,j}_{\frac{z}{2}+1,2}} \succ V^{i,j}_{\frac{z}{2}+1} \succ [ \cdots ]\\
 s^{i,j}_{\frac{z}{2}+1,2} &: \blue{t^{i,j}_{\frac{z}{2}+1,2}} \succ w_j \succ\red{t^{i,j}_{\frac{z}{2}+1,3}} \succ V^{i,j}_{\frac{z}{2}+1} \succ [ \cdots ]\\[-0.5ex]
 &\vdots\\[-0.5ex]
 s^{i,j}_{\frac{z}{2}+1,y} &: \blue{t^{i,j}_{\frac{z}{2}+1,y}} \succ w_j \succ\red{t^{i,j}_{\frac{z}{2}+2,1}} \succ V^{i,j}_{\frac{z}{2}+1} \succ [ \cdots ]\\[-0.5ex]
 &\vdots \\[-1ex]
  &\vdots \\[-1ex]
 s^{i,j}_{z,y-1} &: \blue{t^{i,j}_{z,y-1}} \succ w_j \succ\red{t^{i,j}_{z,y}} \succ V^{i,j}_{z} \succ [ \cdots ]\\
 s^{i,j}_{z,y} &: \blue{t^{i,j}_{z,y}} \succ w_j \succ\red{t^{i,j}_{\frac{z}{2}+1,1}} \succ V^{i,j}_{z} \succ [ \cdots ]\\[-0.5ex]
 p^{i,j}_{1,1} &: v^{i,j}_{1,1} \succ [ \cdots ]\\
 &\vdots\\[-0.5ex]
  &\vdots\\[-0.5ex]
 p^{i,j}_{z,y} &: v^{i,j}_{z,y} \succ [ \cdots ]
\end{alignat*} 
\end{minipage}
 \vrule{} 
\begin{minipage}{0.45\textwidth}
\centering \underline{Women}
\begin{alignat*}{1}
 w_i &: M_{A_1} \succ S \succ M_{A_2}\\
 t^{i,j}_{1,1} &: M_A \succ s^{i,j}_{\frac{z}{2}, y} \succ S_1^{i,j} \succ [ \cdots ]\\
 t^{i,j}_{1,2} &: M_A \succ s^{i,j}_{\frac{z}{2}+1, 1} \succ S_1^{i,j} \succ [ \cdots ]\\
 t^{i,j}_{1,3} &: M_A \succ s^{i,j}_{1, 2} \succ S_1^{i,j} \succ [ \cdots ]\\[-0.5ex]
 &\vdots \\[-0.5ex]
 t^{i,j}_{1,y} &: M_A \succ s^{i,j}_{1, y-1} \succ S_1^{i,j} \succ [ \cdots ]\\
 t^{i,j}_{2,1} &: M_A \succ s^{i,j}_{1, y} \succ S_2^{i,j} \succ [ \cdots ]\\[-0.5ex]
 & \vdots\\[-0.5ex]
 t^{i,j}_{2,y} &: M_A \succ s^{i,j}_{2, y-1} \succ S_2^{i,j} \succ [ \cdots ]\\[-0.5ex]
 &\vdots \\[-0.5ex]
 t^{i,j}_{\frac{z}{2}+1,1} &: M_A \succ s^{i,j}_{1, 1} \succ S_{\frac{z}{2}+1}^{i,j} \succ [ \cdots ]\\
  t^{i,j}_{\frac{z}{2}+1,2} &: M_A \succ s^{i,j}_{\frac{z}{2}+1, 1} \succ S_{\frac{z}{2}+1}^{i,j} \succ [ \cdots ]\\[-0.5ex]
 &\vdots \\[-1ex]
  &\vdots \\[-1ex]
 t^{i,j}_{z,y} &: M_A \succ s^{i,j}_{z, y-1} \succ S_z^{i,j} \succ [ \cdots ]\\
 v^{i,j}_{1,1} &: M_A \succ S_1^{i,j} \succ p^{i,j}_{1,1} \succ [ \cdots ]\\[-0.5ex]
 &\vdots \\[-0.5ex]
 v^{i,j}_{1,y} &: M_A \succ S_1^{i,j} \succ p^{i,j}_{1,y} \succ [ \cdots ]\\
 v^{i,j}_{2,1} &: M_A \succ S_2^{i,j} \succ p^{i,j}_{2,1}\succ [ \cdots ]\\[-0.5ex]
 & \vdots\\[-0.5ex] 
 v^{i,j}_{2,y} &: M_A \succ S_2^{i,j} \succ p^{i,j}_{2,y} \succ [ \cdots ]\\[-0.5ex]
 &\vdots \\[-1ex]
  &\vdots \\[-1ex]
 v^{i,j}_{z,y} &: M_A \succ S_z^{i,j} \succ p^{i,j}_{z,y} \succ [ \cdots ]
\end{alignat*}
\end{minipage}
\noindent
\hrule
\caption{\small The instance $\mathcal{I'}$ that is used in the proof of Theorem~\ref{thm:inapprox-minss}}
\label{fig2}
}
\end{figure} 

Note that $n = |U| = |W| = k + 2yz|E|$. Also, the amount of missing information per agent is at most $\frac{\binom{k}{2} + \binom{y}{2}}{{\binom{n}{2}}} \leq 
\frac{\binom{k+y}{2}}{{\binom{n}{2}}}$. Therefore, the total amount of missing information is $ \leq \frac{\binom{k+y}{2}}{{\binom{n}{2}}} \leq 
\frac{\binom{2y}{2}}{{\binom{n}{2}}} \leq \frac{4y^2}{n^2} \leq \frac{4y^2}{4y^2z^2} \leq \delta$.
Next, in order to show the correctness, we prove the following claims. Again, as noted above, the proofs of these use 
similar ideas as in Hamada et.\ al's proof of Theorem 1 in \citep{hamada16}. Also, throughout, we make use of the following definition: for every $m_i$, if $m_i$ is not 
matched to a woman in $W_A$, then we call such a pair as a \textit{bad pair}. Additionally, for every $s^{i,j}_{a,b}$, if $s^{i,j}_{a,b}$ is matched to a woman who is 
outside of the top three women in his list (i.e., for instance, if $s^{i,j}_{1,1}$ is matched with anyone other than $t_{1,1}^{i,j}$, $w_i$, or $t_{\frac{z}{2}+1,1}^{i,j}$), 
then 
we again call such pairs as being \textit{bad}. 

\begin{claim} \label{clm:probhited-pair}
 If a matching $\mathcal{M}$ contains a bad pair, then it has at least $y-1$ super-blocking pairs. 
\end{claim}

\begin{proof}
 Consider the case when $m_i$ is matched to a woman $w'$ who is not in $W_A$. This implies that it at least forms a super-blocking pair with all $w \in V_1^{i,j}$ such 
that 
$w \neq w'$. And since every women in $V_1^{i,j}$ finds all the men in $M_A$ as incomparable, therefore we at least have $|V_1^{i,j}| - 1 = y -1$ super-blocking 
pairs. 

Next, consider the case when there is a man $s^{i,j}_{a,b}$ who is matched to a woman $w'$ who is outside of the top three women in his list.  This implies that it at 
least 
forms a super-blocking pair with all $w \in V_a^{i,j}$ such that $w \neq w'$. And since we can assume that no women in $V_a^{i,j}$ is matched to a man in $M_A$ (as this 
would anyway result in $y-1$ super-blocking pairs as proved above) and since all of them find the men in $S^{i,j}_a$ as incomparable, this implies that we have at 
least $|V_a^{i,j}| - 1 = y -1$ super-blocking pairs. 
\end{proof}

Before we go on to the next claim, for every $i < j$ such that $(v_i, v_j) \in E$, consider the sets $S^{i,j}$ and $T^{i,j}$ and let us the define the following two perfect 
matchings, $\mathcal{M}^{i,j}_1$ and $\mathcal{M}^{i,j}_2$, between $S^{i,j}$ and $T^{i,j}$. The matching $\mathcal{M}^{i,j}_1$ ($\mathcal{M}^{i,j}_2$ ) can be inferred 
from Figure~\ref{fig2} by matching every man in $S^{i,j}$ with the woman coloured red (blue) in his list. 
\begin{align*}
\mathcal{M}^{i,j}_1 &= \left\{(s_{1,1}^{i,j}, \red{t^{i,j}_{1,1}}), (s_{1,2}^{i,j}, \red{t^{i,j}_{1,2}}), \cdots, (s_{\frac{z}{2},y}^{i,j}, \red{t^{i,j}_{\frac{z}{2},y}}), 
(s_{\frac{z}{2} + 1, 1}^{i,j}, \red{t^{i,j}_{\frac{z}{2}+1,2}}), (s_{\frac{z}{2} + 1, 2}^{i,j}, \red{t^{i,j}_{\frac{z}{2}+1,3}}), \cdots, \allowbreak (s_{z,y}^{i,j}, 
\allowbreak \red{t^{i,j}_{\frac{z}{2}+1,1}}) 
\right\} \\ 
\mathcal{M}^{i,j}_2 &= \left\{ (s_{1,1}^{i,j}, \blue{t^{i,j}_{\frac{z}{2}+1,1}}), (s_{1,2}^{i,j}, \blue{t^{i,j}_{1,3}}), \cdots, (s_{\frac{z}{2},y}^{i,j}, 
\blue{t^{i,j}_{1,1}}), (s_{\frac{z}{2} + 1, 1}^{i,j}, \blue{t^{i,j}_{1,2}}), (s_{\frac{z}{2} + 1, 2}^{i,j}, \blue{t^{i,j}_{\frac{z}{2}+1,2}}), \cdots, \allowbreak 
(s_{z,y}^{i,j}, \allowbreak \blue{t^{i,j}_{z,y}}) \right\} 
\end{align*}

\begin{claim} \label{clm:2matchings}
 For every $i < j$ such that $(v_i, v_j) \in E$, $\mathcal{M}^{i,j}_1$ and $\mathcal{M}^{i,j}_2$ are the only perfect matchings between $S^{i,j}$ and $T^{i,j}$ that do not 
include a bad pair. Moreover, both $\mathcal{M}^{i,j}_1$ and $\mathcal{M}^{i,j}_2$ have only one super-blocking pair $(m, w)$ such that $m\in S^{i,j}$ and $w\in T^{i,j}$. 
\end{claim}

\begin{proof}[Proof (sketch)]
 It is easy to observe the first part. As for the second part, note that none of $s_{1,1}^{i,j}, \cdots, s_{\frac{z}{2}, y}^{i,j}$ form any super-blocking pair in 
$\mathcal{M}_1$ as they are matched to their 
topmost choices. Also, none of $s_{\frac{z}{2}+1,2}^{i,j}, \cdots, s_{z, y}^{i,j}$ form any super-blocking pairs since the only woman they can form super-blocking pairs 
with, 
which are $t_{\frac{z}{2}+1,2}^{i,j}, \cdots t_{z, y}^{i,j}$ respectively, strictly prefers their currently matched partner, which are $s_{\frac{z}{2}+1,2}^{i,j}, 
\allowbreak \cdots, s_{z, y}^{i,j}$ respectively. Hence, the only super-blocking pair is $(s_{\frac{z}{2}+1,1}^{i,j}, t_{1, 2}^{i,j})$.  We can make similar arguments 
with 
respect to $\mathcal{M}_2$ to show that 
 $(s_{1,1}^{i,j}, t_{1, 1}^{i,j})$ is the only super-blocking pair. 
\end{proof}

Given the two claims above, we can now prove the correctness of the reduction through the following lemmas.  

\begin{lemma} \label{lem:yes}
 If $\mathcal{I} = (G, k_0)$ is a ``yes'' instance of VC, then $\mathcal{I}'$ has a solution with at most $2k^2$ super-blocking pairs. 
\end{lemma}

\begin{proof}
 Let the vertex cover of $G$ be $C$ and since it is a ``yes'' instance, we know that $|C| \leq k_0$. If the size of $C$ is strictly less than $k_0$, then add arbitrary 
vertices to it in order to make its size $k_0$. So from now on we can assume that $|C| = k_0$. Next, construct the following matching $\mathcal{M}$ for the instance 
$\mathcal{I}'$.
\begin{itemize}
 \item For every woman $w_i \in W_A$,  if $v_i \in C$, then match $w_i$ with some man in $M_{A_1}$. Otherwise, match $w_i$ with some man in $M_{A_2}$. 
 \item For every $i < j$ such that $(v_i, v_j) \in E$, if $v_i \in C$, then match every man in $S^{i, j}$ with a woman in $T^{i,j}$ using $\mathcal{M}^{i,j}_2$ as defined 
above. Otherwise, match every man in $S^{i, j}$ with a woman in $T^{i,j}$ using $\mathcal{M}^{i,j}_1$ as defined above.  
\item For every $i < j$ such that $(v_i, v_j) \in E$, match $p^{i,j}_{a,b}$ with $v^{i,j}_{a,b}$. 
\end{itemize}

Now, we know that the each man in $M_A$ can form at most $k$ super-blocking pairs (one with each woman in $W_A$). Additionally, we know from Claim~\ref{clm:2matchings} 
that 
both $\mathcal{M}^{i,j}_1$ and $\mathcal{M}^{i,j}_2$ have at most one super-blocking pair, and that none of the men in $P$ form any super-blocking pair as they all get 
their 
topmost choice. Hence, the total number of blocking pairs is at most $k^2 + |E| \leq 2k^2$.  
\end{proof}

\begin{lemma} \label{lem:no}
 If $\mathcal{I} = (G, k_0)$ is a ``no'' instance of VC, then every matching for $\mathcal{I}'$ has at least $y - 1$ super-blocking pairs. 
\end{lemma}

\begin{proof}
Here we will show that if there exists a matching $\mathcal{M}$ with less than $y-1$ super-blocking pairs for $\mathcal{I}'$, then $\mathcal{I}$ has  a 
vertex cover of size at most $k_0$. To see this, consider $\mathcal{M}$. Since it has less than $y-1$ blocking pairs, we know from Claim~\ref{clm:probhited-pair} that it 
does 
not have any bad pair. This in turn implies that all the men in $M_A$ are matched with a woman in $W_A$ (since all men in $M_A$ have to be attached to a woman in $W_A$ 
and size of both the sets are equal). 

Next, for every $i < j$ such that $(v_i, v_j) \in E$, let us consider the men and women in $S^{i,j}$ and $T^{i,j}$. Since, again, we cannot have any bad pairs, we 
know 
that there has to be a perfect matching between these two sets. Additionally, from Claim~\ref{clm:2matchings} we know that $\mathcal{M}^{i,j}_1$ and $\mathcal{M}^{i,j}_2$ 
are 
the only two perfect matchings that have no bad pairs. Now, for an $(i, j)$, if we were using $\mathcal{M}^{i,j}_1$, then it is easy to see that $w_j$ should 
be matched with a man in $M_{A_1}$ as otherwise she would form a super-blocking pair with all the men in $\{s^{i,j}_{\frac{z}{2}, 1}, \cdots, s^{i,j}_{z,y}\}$, thus 
resulting 
in at least $zy > y-1$ super-blocking pairs for $\mathcal{M}$. Similarly, if we were using $\mathcal{M}^{i,j}_2$, then $w_i$ should be matched with a man in $M_{A_1}$ as 
otherwise we would have at least $zy > y-1$ blocking pairs. Therefore, we have that for each edge $(v_i, v_j) \in E$ at least one of the women $w_i$ or $w_j$ should be 
matched 
to a man in $M_{A_1}$. So, now, if we define $C = \{v_i \: | \: \mathcal{M}(w_i) \in M_{A_1}\}$, then we have a vertex cover of size at most $k_0$ (as size of 
$M_{A_1}$ is $k_0$).
\end{proof}

Finally, from Lemmas~\ref{lem:yes} and~\ref{lem:no}, we have an inapproximability gap of $\alpha$, where
\begin{align*}
 \alpha &\geq \frac{y-1}{2k^2} \\
 &= \frac{k^d}{2k^2} \\
 &> \frac{n\sqrt{\delta}}{16k^4} \tag{\text{using the fact that $n = 2yzk^2 + k \leq 8k^{d+2}\frac{1}{\sqrt{\delta}}$}}\\
 &> \left(n\sqrt{\delta}\right)^{1 - \epsilon} \tag{\text{using the fact that $n = 2yzk^2 + k > 2yz > 2k^d\frac{1}{\sqrt{\delta}}$}}
\end{align*}
\end{proof}

\subsection{\texorpdfstring{A possible general approach for obtaining a near-tight approximation factor for $\delta$-min-bp-super-stable-matching}{}} \label{sec:approach}

While obtaining a general near-tight approximation result for the $\delta$-min-bp-super-stable-matching problem is still open, in this section we propose a potentially 
promising direction for this problem. In particular, we demonstrate how solving even a very relaxed version of the min-delete-stable-matching problem will be enough to get an 
$\mathcal{O}(n)$-approximation for $\delta$-min-bp-super-stable-matching in general. Below, we first define the relaxation in question, which we refer to as an $(\alpha, 
\beta)$-approximation to the min-delete-super-stable-matching problem. 

\begin{definition}[$(\alpha, \beta)$-min-delete-super-stable-matching] Given an instance $\mathcal{I} = (\delta, p_U, p_W)$, compute a set $D'$ such that $|D'| \leq \alpha 
\cdot |D_{opt}|$, where $|D_{opt}|$ 
is the size of the optimal solution to the min-delete-super-stable-matching for the same instance, and the instance $\mathcal{I}_{-D'} = (\delta_{-D'}, p_{U\setminus 
D'}, p_{W\setminus D'})$, where $\delta_{-D'} = \frac{1}{|(U \cup W) \setminus D'|}\sum_{i \in (U \cup W) \setminus D'} \delta_i$, has a matching with at most $\beta$ 
super-blocking pairs. 
\end{definition}

Next, we show that an $(\alpha, \beta)$-approximation to the min-delete-super-stable-matching problem gives us an $(\alpha n + \beta)$-approximation for 
$\delta$-min-bp-super-stable-matching. So, in particular, if we have an $(\alpha, \beta)$-approximation where $\alpha$ is a constant and $\beta \in \mathcal{O}(n)$, 
then this in turn gives us an $\mathcal{O}(n)$-approximation for $\delta$-min-bp-super-stable-matching in general. 

\begin{proposition} \label{prop:abmindel}
 If there exists an $(\alpha, \beta)$-approximation algorithm for the min-delete-super-stable-matching problem, then there exists an $(\alpha n + \beta)$-approximation 
algorithm for the $\delta$-min-bp-super-stable-matching problem.
\end{proposition}

\begin{proof}[Proof (sketch)]
We can proceed to prove this almost exactly as in the proof of Theorem~\ref{thm:napprox}. Here, if $D$ denotes the $(\alpha, \beta)$-approximate solution returned by the 
algorithm, then the only difference is that we define $\mathcal{M}_1$ to be the matching with the set of agents in $(U \cup W) \setminus D$ such that it has at most 
$\beta$ super-blocking pairs (from the definition of the problem we know that such a matching exists) and $\mathcal{M}_2$ to be an arbitrary matching on the set of 
agents in $D$. Once we have this, then we can arrive at the bound by proceeding exactly as in the proof of Theorem~\ref{thm:napprox}, with the only difference being that 
here 
we would use $S_1$, which is the number of super-blocking pairs associated with $\mathcal{M}_1$, to be equal to $\left(n - \frac{|D|}{2}\right)\cdot\frac{|D|}{2} + \beta$. 
\end{proof}

\section{Conclusion} \label{sec:conclude}
In this paper we initiated a study on matching with partial information in order to investigate what makes a matching ``good'' in this context, and to better understand the 
trade-off between the amount of missing information and the quality of different matchings. Towards this end, we introduced a measure for accounting for missing preference 
information in an instance, and argued that a natural definition of a ``good'' matching in this context is one that minimizes the maximum number of blocking pairs 
with respect to all the possible completions. Subsequently, using an equivalent problem 
($\delta$-min-bp-super-stable-matching) we first explored the space of matchings that contained no obvious blocking pairs (i.e., weakly-stable matchings) in order to better 
understand how missing preference information effected/affected the quality, in terms of approximation with respect to the objective of minimizing the number of 
super-blocking pairs. Later on, by expanding the space of matchings we considered (i.e., removing the restriction that matches must be weakly-stable), we asked whether it was 
possible to improve on the approximation factors that were achieved under the restriction to weakly-stable matchings.  

There are a number of interesting directions for future work.  First, while in Section~\ref{sec:approach} we proposed one possible approach that can lead to near-tight 
approximations, there may be other approaches that can prove fruitful. Second, we believe that the min-delete-super-stable-matching problem, and its 
relaxation 
we introduced, are both of independent interest, and so an open question is to see if one can obtain general results on them. In Proposition~\ref{prop:2approx} we 
saw that a 2-approximation was achievable for the case of one-sided top-truncated preferences and hence it would also be interesting to determine if there are other 
interesting classes of preferences for which constant-factor approximations are possible. Finally, there are possible extensions, like, for instance, allowing 
incompleteness---meaning the agents can specify that they are willing to be matched to only a subset of the agents on the other set---that one could consider and ask similar 
questions like the ones we considered.

\printbibliography

\appendix

\section{Example to illustrate ``bad'' weakly-stable matchings in the case of one-sided top-truncated preferences} \label{sec:example}

Consider the instance $\mathcal{I}$ as shown in Figure~\ref{fig4}, where ties appear only on the women's side. Furthermore, we define the following: 

\begin{figure}[t!] 
{\scriptsize
\noindent\hrulefill\\[0.5ex]
\begin{minipage}{0.49\textwidth}
\centering \underline{Men}
\begin{alignat*}{1}
 m_1 &: w_1 \succ W_{F\setminus\{1\}} \succ W_{B_1}\succ \cdots \succ W_{B_z}\\
 m_2 &: w_1 \succ w_2 \succ [ \cdots ]\\
 m_3 &: w_2 \succ w_3 \succ W_{F\setminus\{2, 3\}} \succ W_S\\
 m_4 &: w_2 \succ w_4 \succ W_{F\setminus\{2, 4\}} \succ W_S\\[-1ex]
 &\hspace{20mm}\vdots\\[-1ex]
 m_{\frac{n}{2}} &: w_2 \succ w_{\frac{n}{2}} \succ W_{F\setminus\{2, \frac{n}{2}\}} \succ W_S\\
 m_{b_0} &: w_1 \succ W_{B_1 \setminus \{b_0\}} \succ w_{b_0} \succ W_{S\setminus B_1}  \succ W_{F\setminus\{1\}}\\[-1ex]
 &\hspace{20mm}\vdots\\[-1ex]
 m_{b_1 - 1} &: w_1 \succ W_{B_1 \setminus \{b_1 - 1\}} \succ w_{b_1 - 1} \succ W_{S\setminus B_1}  \succ W_{F\setminus\{1\}}\\
 m_{b_1} &: w_1 \succ W_{B_2 \setminus \{b_1\}} \succ w_{b_1} \succ  W_{S\setminus B_2} \succ W_{F\setminus\{1\}}\\[-1ex]
 &\hspace{20mm}\vdots\\[-1ex]
 m_{b_2 - 1} &: w_1 \succ W_{B_2 \setminus \{b_2 - 1\}} \succ w_{b_2 - 1} \succ  W_{S\setminus B_2} \succ W_{F\setminus\{1\}}\\[-1ex]
 &\hspace{20mm}\vdots\\[-1ex]
   m_{b_{z-1}} &: w_1 \succ W_{B_z \setminus \{b_{z-1}\}}  \succ  w_{b_{z-1}} \succ W_{S\setminus B_z} \succ W_{F\setminus\{1\}}\\[-1ex]
 &\hspace{20mm}\vdots\\[-1ex]
 m_{b_{z} - 1} &: w_1 \succ W_{B_z \setminus \{b_{z} - 1\}} \succ w_{b_{z} - 1} \succ W_{S\setminus B_z} \succ  W_{F\setminus\{1\}}
\end{alignat*} 
\end{minipage}
 \vrule{} 
\begin{minipage}{0.49\textwidth}
\centering \underline{Women}
\begin{alignat*}{1}
 w_1&: m_2 \succ m_1 \succ [ \cdots ]\\
 w_2&: m_2 \succ m_1 \succ [ \cdots ]\\
 w_3&: m_1 \succ m_3 \succ [ \cdots ]\\
 w_4&: m_1 \succ m_4 \succ [ \cdots ]\\[-1ex]
 &\hspace{20mm}\vdots\\[-1ex]
 w_{\frac{n}{2}}&: m_1 \succ m_{\frac{n}{2}} \succ [ \cdots ]\\
 w_{b_0}&: M_{S\setminus B_1} \succ m_1 \succ m_{b_0} \succ M_{F\setminus\{1\}} \succ M^T_{B_1 \setminus \{b_0\}}\\[-1ex]
 &\hspace{20mm}\vdots\\[-1ex]
 w_{b_1 - 1} &:  M_{S\setminus B_1} \succ m_1 \succ m_{b_1-1} \succ M_{F\setminus\{1\}} \succ M^T_{B_1 \setminus \{b_1 - 1\}} \\
 w_{b_1} &:  M_{S\setminus B_1} \succ m_1 \succ m_{b_1} \succ M_{F\setminus\{1\}} \succ M^T_{B_1 \setminus \{b_1\}}\\[-1ex]
 &\hspace{20mm}\vdots\\[-1ex]
 w_{b_2 - 1} &:  M_{S\setminus B_1} \succ m_1 \succ m_{b_2-1} \succ M_{F\setminus\{1\}} \succ M^T_{B_1 \setminus \{b_2-1\}}\\[-1ex]
 &\hspace{20mm}\vdots\\[-1ex]
 w_{b_{z-1}} &:  M_{S\setminus B_1} \succ m_1 \succ m_{b_{z-1}} \succ M_{F\setminus\{1\}} \succ M^T_{B_1 \setminus \{b_{z-1}\}}\\[-1ex]
 &\hspace{20mm}\vdots\\[-1ex]
 w_{b_{z} - 1} &: M_{S\setminus B_1} \succ m_1 \succ m_{b_z-1} \succ M_{F\setminus\{1\}} \succ M^T_{B_1 \setminus \{b_z-1\}}
\end{alignat*}
\end{minipage}
\hrule 
\caption{\small The instance $\mathcal{I}$ that is used to illustrate that there can be weakly-stable matchings with $\mathcal{O}(n^2\sqrt{\delta})$ super-blocking pairs even 
in the case of 
one-sided top-truncated preferences}
\label{fig4}
}
\end{figure}

\begin{itemize}
\item $\delta \in [\frac{16}{n^2}, \frac{1}{4}]$,
  $y = \frac{n\sqrt{\delta}}{2}$,  $z = \frac{n}{2y}$ (for simplicity we assume that $y$ and $z$ are integers; we can appropriately modify the proof if that is not the 
case)
 \item $b_j = \frac{n}{2} + jy + 1, \forall j \in [0, \cdots z]$
 \item$B_i = \{b_{i-1}, \cdots, b_{i} - 1\}, \forall i \in [1, \cdots z]$, 
 $F = \{1, \cdots, \frac{n}{2}\}, S = \{\frac{n}{2} + 1, \cdots, n\}$
   \item$W_X\,(M_X):$ {for some set $X$, place all the women (men) with index in $X$ in the increasing order of their indices}
 \item$W^T_X\,(M^T_X):$ {for some set $X$, place all the women (men) with index in $X$ as tied}
 \item$[ \cdots ]:$ {place all the remaining alternatives in some strict order}.
\end{itemize}

First thing is to see that the optimal solution $\mathcal{M}_{opt}$ associated with the instance is
\begin{equation*}
 \mathcal{M}_{opt} = \left\{(m_1, w_1), (m_2, w_2), \cdots, (m_n, w_n)\right\},
\end{equation*}
where $(m_2, w_1)$ is the only super-blocking pair (and it is an obvious blocking pair). Also, it can be verified that the total amount of missing information in $\mathcal{I}$ 
is at most $\delta$.

Now, consider the matching $\mathcal{M}$, where 
\begin{multline*}
 \mathcal{M} = \Bigg\{(m_1, w_2), (m_2, w_1), (m_3, w_3), (m_4, w_4) \cdots, (m_{\frac{n}{2}}, w_{\frac{n}{2}}), \\ (m_{b_0}, w_{b_0+1}), (m_{b_0+1}, w_{b_0+2}), \cdots, 
(m_{b_1 
- 2}, w_{b_1-1}), (m_{b_1 - 1}, w_{b_0}), \\ (m_{b_1}, w_{b_1+1}), (m_{b_1+1}, w_{b_1+2}), \cdots, (m_{b_2 - 2}, w_{b_2-1}), (m_{b_2 - 1}, w_{b_1}), \cdots,\\ \cdots, 
(m_{b_{z-1}}, w_{b_{z-1}+1}), (m_{b_{z-1}+1}, w_{b_{z-1}+2}), \cdots, (m_{b_{z}-2}, w_{b_{z}-1}), (m_{b_{z}-1}, w_{b_{z-1}})\Bigg\}.
\end{multline*}
It is easy to check that $\mathcal{M}$ is weakly-stable. Also, it can be verified that it has $\mathcal{O}(n^2\sqrt{\delta})$ super-blocking pairs (this is because 
with respect to each block $B_j$ one can see that $\mathcal{M}$ has $\mathcal{O}\left(|B_j|^2\right)$ super-blocking pairs).\qed

\end{document}